\documentclass[sigconf]{aamas_hacked}

\usepackage{balance} \usepackage{thm-restate}
\usepackage{mathtools}
\usepackage{tikz}
\usetikzlibrary{positioning,arrows,shapes,decorations,calc,fit,arrows.meta,colorbrewer} \usepackage{xspace}
\usepackage{enumitem}
\usepackage{todonotes}
\usepackage{cleveref}

\usetikzlibrary{datavisualization}

\makeatletter
\gdef\@copyrightpermission{
  \begin{minipage}{0.2\columnwidth}
   \href{https://creativecommons.org/licenses/by/4.0/}{\includegraphics[width=0.90\textwidth]{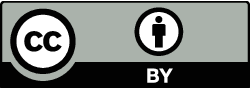}}
  \end{minipage}\hfill
  \begin{minipage}{0.8\columnwidth}
   \href{https://creativecommons.org/licenses/by/4.0/}{This work is licensed under a Creative Commons Attribution International 4.0 License.}
  \end{minipage}
  \vspace{5pt}
}
\makeatother

\setcopyright{ifaamas}
\acmConference[AAMAS '26]{Proc.\@ of the 25th International Conference
on Autonomous Agents and Multiagent Systems (AAMAS 2026)}{May 25 -- 29, 2026}
{Paphos, Cyprus}{C.~Amato, L.~Dennis, V.~Mascardi, J.~Thangarajah (eds.)}
\copyrightyear{2026}
\acmYear{2026}
\acmDOI{}
\acmPrice{}
\acmISBN{}

\title{Majoritarian Assignment Rules}

\author{Felix Brandt}
\affiliation{
  \institution{Technical University of Munich}
\city{Munich}
  \country{Germany}}
\email{brandt@tum.de}

\author{Haoyuan Chen}
\affiliation{
  \institution{Technical University of Munich}
  \city{Munich}
  \country{Germany}}
\email{ge85hud@tum.de}

\author{Chris Dong}
\affiliation{
  \institution{HPI, University~of~Potsdam}
  \city{Potsdam}
  \country{Germany}}
\email{chrisshuyu.dong@hpi.de}

\author{Patrick Lederer}
\affiliation{
  \institution{ILLC, University of Amsterdam}
  \city{Amsterdam}
  \country{The Netherlands}}
\email{p.lederer@uva.nl}

\author{Alexander Schlenga}
\affiliation{
  \institution{Technical University of Munich}
  \city{Munich}
  \country{Germany}}
\email{alexander.schlenga@tum.de}

\begin{abstract}
A central problem in multiagent systems is the fair assignment of objects to agents. In this paper, we initiate the analysis of classic majoritarian social choice functions
in assignment. Exploiting the special structure of the assignment domain, we show a number of surprising results with no counterparts in general social choice.
In particular, we establish a near one-to-one correspondence between preference profiles and majority graphs. This correspondence implies that key properties of assignments---such as Pareto-optimality, least unpopularity, and mixed popularity---can be determined solely by the associated majority graph.
We further show that all Pareto-optimal assignments are semi-popular and belong to the top cycle.
Elements of the top cycle can thus easily be found via serial dictatorships.
Our main result is a complete characterization of the top cycle, which implies the top cycle can only consist of one, two, all but two, all but one, or all assignments. 
By contrast, we find that the uncovered set contains only very few assignments.\end{abstract}

\keywords{House Allocation; Social Choice Theory; Top Cycle; Uncovered Set}

\newcommand{\BibTeX}{\rm B\kern-.05em{\sc i\kern-.025em b}\kern-.08em\TeX}
\newcommand{\pref}{\succ} 
\newcommand{\prefsim}{\succsim}
\newcommand{\dispref}{\prec}

\newcommand{\mpref}{\pref}
\newcommand{\mprefsim}{\prefsim}
\newcommand{\mprefsimt}{\mprefsim^*}

\newcommand{\order}{\sigma}
\newcommand{\goodass}[1]{\textcolor{blue}{#1}} \newcommand{\badass}[1]{\textcolor{red}{#1}} \newcommand{\bc}{\ensuremath{\mathit{BC}}\xspace}
\newcommand{\pso}{\ensuremath{\mathit{PP}}\xspace}

\newcommand{\po}{\ensuremath{\mathit{PO}}\xspace}

\newcommand{\tc}{\ensuremath{\mathit{TC}}\xspace}
\newcommand{\uc}{\ensuremath{\mathit{UC}}\xspace}

\theoremstyle{definition}

\newtheorem{example}{Example}
\theoremstyle{plain}
\newtheorem{theorem}{Theorem}
\newtheorem{lemma}{Lemma}

\newtheorem{corollary}{Corollary}

\newtheorem{fact}{Fact}

\newtheoremstyle{bfnote}{}{}{}{}{\bfseries}{.}{ }{\thmname{#1}\thmnumber{ #2}\thmnote{\textnormal{ (#3)}}}
\theoremstyle{bfnote}
\newtheorem{remark}{Remark}

\begin{document}

\pagestyle{fancy}

\maketitle

\section{Introduction}

Assigning objects to individual agents is a fundamental problem that has received considerable attention by computer scientists as well as economists \citep[e.g.,][]{CDE+06a,SoUn10a,Manl13a,BCM15a}. The problem is known as the \emph{assignment problem}, the \emph{house allocation problem}, or \emph{two-sided matching with one-sided preferences}. In its simplest form, there are $n$ agents, $n$ houses, and each house needs to be allocated to exactly one agent based on the strict preferences of each agent over the houses.
Applications are diverse and include assigning jobs to applicants, apartments to tenants, and offices to employees.

A natural way to compare two assignments $\mu$ and~$\lambda$ in such settings is to check whether a majority of agents prefer the house they receive under~$\mu$ to the one they receive under~$\lambda$. This idea leads to the notion of \emph{popular} assignments, which are assignments for which there is no other assignment that is preferred by a strict majority of the agents \citep{Gard75a}. However, as \citet{Gard75a} pointed out, such popular assignments may not exist, as the majority relation over assignments can be cyclic.
This observation has led to an extensive stream of follow-up work \citep[e.g.,][]{AIKM07a,McIr11a,ABS13a,CHK15a,Cseh17a,BHS17a,BrBu20a}, which aims to understand under which conditions popular assignments exist and which assignment should be selected in their absence. For example, the latter question has led to the definition of concepts such as \emph{least unpopular} or \emph{mixed popular} assignments \citep{McCu08a,KMN11a}.

Curiously, this line of work closely resembles classic ideas in \emph{social choice}. In this setting, agents report preferences over some alternatives and a social choice function (SCF) returns the ``best'' alternative based on these preferences. The observation that popular assignments do not always exist is analogous to the \emph{Condorcet paradox} in social choice, which states that (weak) Condorcet winners need not exist. Even more, these two insights are logically related as house allocation can be viewed as a special case of {social choice}. When letting the set of alternatives be the set of all possible allocations of houses to agents and postulating that agents are indifferent between all allocations in which they receive the same house, house allocation is reduced to a restricted domain of general social choice. Under this reduction, a popular assignment is merely a weak Condorcet winner, so the absence of a popular assignment shows the Condorcet paradox for a restricted domain of preferences.\footnote{Majority cycles are even more frequent in assignment than they are in social choice, stressing the importance of which assignment to choose in the absence of popular ones. For example, this can be seen by comparing the proportions of profiles that admit Condorcet winners and popular assignments, respectively \citep[][Tables 1 and 4.2]{GeFi79a,AIKM07a}.}
Moreover, several assignment concepts based on popularity have counterparts in social choice: for example, the notion of least popularity corresponds to the \emph{maximin} voting rule and mixed popularity to \emph{maximal lotteries}.

Motivated by this relationship, we aim to transfer further ideas from social choice, namely \emph{majoritarian SCFs}, to the assignment domain.
 Just like popularity, such majoritarian SCFs only depend on the majority relation. Typical examples are Copeland's rule, the top cycle, the uncovered set, and the bipartisan set \citep[e.g.,][]{Lasl97a,BBH15a}.
The definitions of these functions are equally natural in assignment as they are in social choice. The top cycle, for example, returns all assignments that are maximal elements of the transitive closure of the majority relation. Similarly, the known set-theoretic relationships between these SCFs also hold in the assignment domain. Computational properties, on the other hand, do \emph{not} carry over from social choice to assignment because, when viewing house allocation as a subdomain of social choice, the number of alternatives is exponential and the agent's preferences are concisely represented as each agent merely cares about her own house.
This has serious algorithmic repercussions, and the computational complexity of even the simplest concepts needs to be reevaluated. For instance, identifying weak Condorcet winners is straightforward in social choice, but finding a popular assignment already requires clever algorithmic techniques~\citep{AIKM07a}.

\subsection*{Contribution}

In this paper, we initiate the systematic study of majoritarian SCFs in the context of assignment. 
To this end, we first characterize all sets of preference profiles that admit the same majority graph by providing an efficient algorithm that reconstructs all profiles that induce a given majority graph. In sharp contrast to the social choice domain, it turns out that almost all majority graphs are induced by a \emph{single} preference profile.
As a consequence, the rules that return all Pareto-optimal assignments, all least unpopular assignments, and all mixed popular assignments, respectively, are majoritarian, even though these concepts are not defined in terms of the majority relation.
Moreover, we show that all Pareto-optimal assignments are contained in the top cycle, which means that elements of the top cycle can be found via serial dictatorships. Further, Pareto-optimal assignments have non-negative Copeland score and are thus semi-popular. None of these results holds in the general social choice domain. 
We also prove that the rule returning all rank-maximal assignments is not majoritarian. 

Our main result is a complete characterization of the top cycle in assignment when there are $n\geq5$ agents and houses. This characterization shows that the cardinality of the top cycle may only take one of five values ($1$, $2$, $n!-2$, $n!-1$, or $n!$) and leads to a simple sublinear-time algorithm that returns a concise representation of all assignments in the top cycle. Once again, this characterization has no analog in the social choice domain, where the top cycle can consist of any number of alternatives.

Lastly, we investigate the three most common variants of the uncovered set, all of which refine the top cycle.
We compute the number of assignments contained in the uncovered set for $n=5$ (by exhaustive enumeration) and for $n=6$ and $n=7$ (by sampling preference profiles). 
Somewhat surprisingly, in all these settings, most profiles only admit \emph{two} uncovered assignments, suggesting that the uncovered set is much more discriminating in assignment than it is in social choice \citep[cf.][]{Fey08a,BrSe15a}.

\subsection*{Related Work}

The study of matching under preferences was initiated by \citet{GaSh62a}. In their model (nowadays often referred to as marriage markets), there are two types of agents who have strict preferences over each other. \citeauthor{GaSh62a} showed that a so-called stable matching always exists and can be found by a simple, efficient algorithm.
By contrast, \Citet{Gard75a} showed that stable matchings are no longer guaranteed to exist when agents may be indifferent between other agents. As a remedy, he proposed the notion of popular matchings, i.e., matchings such that there is no other matching that a majority of the agents prefer. 
When individual preferences are strict, popularity is weaker than stability. However, popular matchings may not exist for weak preferences, even if agents of one type are completely indifferent between all agents of the other type and all other agents have strict preferences.
This variant, which goes back to \citet{Gale60a}, is known as assignment, house allocation, or two-sided matching with one-sided preferences as fully indifferent agents can be seen as objects. \Citet{AIKM07a} provided an efficient algorithm for finding a popular assignment or returning that none exists. 

The definition of popularity has been relaxed in various ways to restore existence. \Citet{McCu08a} proposed the unpopularity margin as a qualitative relaxation of popularity. A least unpopular matching in this sense always exists, but is NP-hard to compute. \Citet{KMN11a} introduced mixed popular assignments, whose existence is guaranteed by the minimax theorem, and provided an efficient algorithm for computing them. \Citet{KaVa23a} transferred the notion of Copeland winners from social choice to so-called roommate markets, which generalize the assignment setting. They showed that computing a Copeland winner is NP-hard. It is open whether this hardness result also holds in assignment. \Citeauthor{KaVa23a} also give a randomized FPTAS for identifying semi-popular matchings, another relaxation of popularity, in roommate markets. Semi-popular matchings are matchings that lose at most half of their majority comparisons
\citep{Kavi20a}. Interestingly, one of our results,
\Cref{thm:PO_Subset_SP}, entails that semi-popular \emph{assignments} can be found efficiently via serial dictatorships in assignment.

Lastly, our work is related to the problem of voting on combinatorial domains \citep[see, e.g.,][]{LaXi15a}. Two recent papers in this line of research that are close to ours are due to \citet{BBS25a} and \citet{BoDi25a}, who study matching through the lens of voting theory.

\section{Preliminaries}

Let $n\in\mathbb{N}$ be given.
We denote by $N = \{1,\dots, n\}$ a set of \emph{agents} and by $H=\{a,b,c,\dots\}$ a set of $n$ \emph{houses} (or distinct indivisible objects in general). Each agent $x\in N$ reports a \emph{preference relation} $\pref_x$, which is formally a linear order over $H$. Intuitively $a\pref_x b$ means that agent $x$ prefers house $a$ to house $b$. Note that we require each agent to rank all houses without indifference and that there are as many houses as agents. A \emph{preference profile} $P=(\pref_1,\dots,\pref_n)$ is the collection of the preference relations of all agents.

Given a preference profile, our goal is to assign one house to each agent. To formalize this, we define \emph{assignments} as bijective functions mapping agents in $N$ to houses in $H$. Thus, $\mu(x)$ is the house given to agent $x$ under assignment $\mu$. We define by $M$ the set of all possible assignments from $N$ to $H$. Further, 
we write an assignment $\mu$ in which agents $1,2,3,\dots$ obtain houses $a,b,c,\dots$, respectively, as $\mu=(a, b, c,\dots)$. Throughout the paper, we assume that agents compare assignments only based on the houses they receive: an agent $x\in N$ (weakly) prefers an assignment $\mu$ to another assignment $\lambda$, denoted by $\mu \succeq_x \lambda$, if $\mu(x)\pref_x\lambda(x)$ or $\mu(x)=\lambda(x)$. Moreover, an agent $x\in N$ strictly prefers an assignment $\mu$ to an assignment $\lambda$, written as $\mu \pref_x\lambda$, if $\mu(x)\pref_x\lambda(x)$ and $\mu(x)\neq\lambda(x)$.

An \emph{assignment rule} $F$ maps every preference profile $P$ to a non-empty set of assignments $F(P)$. 
The idea is that an assignment rule returns a set of ``good'' assignments, from which a single assignment will eventually be picked. 

Most of the rules considered in this paper are \emph{symmetric}, which demands that all agents and all houses are treated equally, respectively. More formally, relabeling the agents in the preference profile should correspond to relabeling the agents in the returned assignments, 
and relabeling the houses in all agents' rankings should correspond to relabeling the houses in the returned assignments. When $n\geq 2$, no assignment rule that always returns a single assignment can be symmetric. Therefore, we consider set-valued rules. 
To nonetheless distinguish more discriminating from less discriminating rules, we say that a rule $F$ is a \emph{refinement} of a rule $G$ and write $F\subseteq G$, if $F(P)\subseteq G(P)$ for all profiles~$P$.

Lastly, we discuss a standard property of assignments called Pareto-optimality. Intuitively, this notion requires that there is no assignment that makes one agent strictly better off without making another one worse off. To formalize this idea, we say an assignment $\mu$ \emph{Pareto-dominates} another assignment $\lambda$ in a profile $P$ if all agents weakly prefer $\mu$ to $\lambda$ and this preference is strict for at least one agent, i.e., $\mu \succeq_x\lambda$ for all agents $x\in N$ and $\mu \succ_x\lambda$ for at least one agent $x\in N$. Further, an assignment is \emph{Pareto-optimal} if it is not Pareto-dominated by any other assignment. 
The set of Pareto-optimal assignments in a profile $P$ is denoted by $\po(P)$. 

The set of Pareto-optimal assignments is closely connected to the family of \emph{serial dictatorships}. Such serial dictatorships are defined by a priority order $\order=(x_1,\dots,x_n)$ over the agents, and agents simply pick their favorite house that has not been taken yet in the order given by $\order$. It is known that an assignment is Pareto-optimal if and only if it is returned by a serial dictatorship for the given profile \citep{AbSo98a}. We note that Pareto-optimality is much more restrictive in assignment than in social choice: there are \emph{always} Pareto-dominated assignments unless all agents have the same preferences. More generally, for every pair of agents $x,y$ and every pair of houses $a,b$ such that $a\pref_x b$ and $b\pref_y a$, there are $(n-2)!$ Pareto-dominated assignments where $x$ obtains $b$ and $y$ obtains $a$.

\section{The Structure of Majority Graphs}

A fundamental way to compare two assignments to each other is to postulate that one assignment is socially preferred to another if a majority of the agents prefer the former to the latter. To this end, let $N_{\mu,\lambda}=\left\{x\in N \mid\mu\succeq_x\lambda\right\}$ denote the set of agents who weakly prefer $\mu$ to $\lambda$. An assignment $\mu$ \emph{weakly majority dominates} assignment $\lambda$ if at least as many agents prefer $\mu$ to $\lambda$ than vice versa, i.e., $\mu\mprefsim\lambda$ if $\lvert N_{\mu,\lambda}\rvert\ge\lvert N_{\lambda,\mu}\rvert$.
Similarly, an assignment $\mu$ \emph{(strictly) majority dominates} another assignment $\lambda$ if strictly more agents prefer $\mu$ to $\lambda$ than vice versa, i.e., $\mu\mpref\lambda$ if $\lvert N_{\mu,\lambda}\rvert>\lvert N_{\lambda,\mu}\rvert$. 

This naturally leads to the analysis of \emph{majority graphs}, which have been extensively studied in social choice theory \citep[e.g.,][]{Lasl97a,BBH15a}. Specifically, the majority graph of a profile $P$ is the directed graph $G_P=(M, \{(\mu,\lambda)\in M^2\colon {\mu \mprefsim \lambda}\})$, which has all possible assignments as its vertices and there is an edge from an assignment $\mu$ to another assignment $\lambda$ if $\mu$ weakly majority dominates $\lambda$.

\begin{example}\label{ex:majority_graph}
Consider the following profile $P$ with $N=\{1,2,3\}$ and $H=\{a,b,c\}$, and the corresponding majority graph.
A black arrow from $\mu$ to $\lambda$ indicates that $\mu$ strictly majority dominates $\lambda$, whereas a bidirectional gray edge indicates a majority tie. 
\vspace{-1ex} 
    \[\arraycolsep=3pt P = \begin{array}{rccc}
		1\colon & a,&b,&c\\
		2\colon & a,&b,&c\\
		3\colon & a,&b,&c\\
		\end{array}
        \quad
\begin{tikzpicture}[baseline=(current bounding box),>=Latex,xscale=1.6]
	\tikzstyle{vertex}=[minimum height={height("$($")+6pt},inner sep=1] 
	\draw
    (0,0.4) node[vertex] (a) {$(a,b,c)$}
    (1,0.4) node[vertex] (b) {$(b,c,a)$}
    (2,0.4) node[vertex] (c) {$(c,a,b)$}
    (0,-0.6) node[vertex] (e) {$(c,b,a)$}
    (1,-0.6) node[vertex] (f) {$(a,c,b)$}
    (2,-0.6) node[vertex] (g) {$(b,a,c)$};

	\draw [->] (a) to (b);
	\draw [->] (b) to (c);
	\draw [->,out=45] (c) to (a);

	\draw [->] (e) to (f);
	\draw [->] (f) to (g);
	\draw [->,out=-45,in=225] (g) to (e);

    \draw [<->,gray!50] (a) to (e);
    \draw [<->,gray!50] (a) to (f);
    \draw [<->,gray!50] (a) to (g);
    \draw [<->,gray!50] (b) to (e);
    \draw [<->,gray!50] (b) to (f);
    \draw [<->,gray!50] (b) to (g);
    \draw [<->,gray!50] (c) to (e);
    \draw [<->,gray!50] (c) to (f);
    \draw [<->,gray!50] (c) to (g);
\end{tikzpicture}\vspace{-1ex}
    \] 
Note that we can obtain the same majority graph from other profiles, too. Specifically, if all three agents $x$ rank $b\pref'_x c\pref'_x a$ in profile $P'$, or all agents rank $c\pref''_x a\pref''_x b$ in profile $P''$, then $G_P = G_{P'} = G_{P''}$.
    
\end{example}

While every directed graph is induced by some preference profile in social choice \citep{McGa53a}, \Citet{BHS17a} pointed out that this is not the case in assignment, where only a small fraction of majority graphs can actually be realized by preference profiles.
Moreover, \citeauthor{BHS17a} gave an efficiently testable, necessary, and sufficient condition for two profiles yielding the same \emph{weighted} majority graph, where each edge $(\lambda,\mu)$ of the majority graph is weighted by the margin $\lvert N_{\lambda,\mu}\rvert - \lvert N_{\mu,\lambda}\rvert $ of the majority comparison. 
{In this section, we will generalize this result to unweighted majority graphs, showing that \emph{almost all} majority graphs are induced by a \emph{unique} profile.}

To this end, we first recall some terminology by \citet{BHS17a}. Let $P$ be a profile and let $(H_1, \dots, H_k)$ be an ordered partition of~$H$, where we call each $H_j$ a \textit{component}. We say that $(H_j)_j$ is a \textit{decomposition} of this profile, if all agents rank all houses in $H_1$ over all houses in $H_2$ and so on.
Formally, for all $j<\ell\le k$ and all $p\in H_j, q\in H_\ell$, and $x\in N$, it holds that $p\pref_x q$. Two profiles $P$, $P'$ are called \textit{rotation equivalent}, if the preferences on $P$ and $P'$ coincide within each component, and one ordering of the components is obtained by shifting the other. 
Formally, consider any decomposition $(H_1,\dots, H_k)$ of ${P}$. Then, for all $j\le k$, $p,q\in H_j$, and $x\in N$, it has to hold that $p\succ_x q$ if and only if $p\succ_x' q$, and there exists some $r< k$ such that $(H_{1+r}, \dots, H_{k+r})$ is a decomposition of $P'$ (where we set $H_{j+r} \coloneqq H_{j+r-k}$ if $j+r>k$).

\begin{example}To illustrate rotation equivalence, consider the following profiles $\hat P, \overline{P}$ with $N=\{1,2,3,4\}$, $H=\{a,b,c,d\}$, and decomposition $(H_1=\{a\}, H_2 = \{b\}, H_3 = \{c,d\})$.

    \[\arraycolsep=3pt \hat{P} = \begin{array}{rcccc}
    1\colon & a,&b,&c,&d\\
    2\colon & a,&b,&c,&d\\
	3\colon & a,&b,&d,&c\\
	4\colon & a,&b,&d,&c\\
    \end{array}
    \qquad
    \overline{P} = \begin{array}{rcccc}
    1\colon & a,&b,&c,&d\\
    2\colon & a,&b,&d,&c\\
	3\colon & a,&b,&d,&c\\
	4\colon & a,&b,&d,&c\\
    \end{array}
    \]

    The profile $\overline{P}$ is \emph{not} rotation equivalent to $\hat{P}$, as $c\mathrel{\hat{\pref}_2}d$ and $d\mathrel{\bar{\pref}_2}c$ even though $c$ and $d$ belong to the same component~$H_3$.
Moreover, rotation equivalence can also be violated when, within each component, the preferences of the agents are coherent.
    For this, consider the profile $P$ from \Cref{ex:majority_graph} and let $P'''$ be the profile, where all agents $x$ report 
    $a\pref'''_x c\pref'''_x b$. Then, $P'''$ is not rotation equivalent to $P$, as the decomposition $(\{a\}, \{b\}, \{c\})$ cannot be rotated to $(\{a\}, \{c\}, \{b\})$. However, $(\{b\}, \{c\}, \{a\})$ and $(\{c\}, \{a\}, \{b\})$ are valid rotations. Hence, the profiles $P'$ and $P''$ described in \Cref{ex:majority_graph} are rotation equivalent to $P$.
\end{example}

{\citet{BHS17a} 
showed that rotation equivalence characterizes the profiles that induce the same \emph{weighted} majority graph. We are able to strengthen this result by showing that the margins are not required: rotation equivalence, in fact, characterizes the profiles inducing the same (unweighted) majority graph! As a consequence, given any assignment-induced majority graph, we can reconstruct the margins of all majority edges.}
Moreover, let a house ``Pareto-dominate'' another house if all agents rank the former above the latter. 
Whenever there are no Pareto-dominated houses in a profile $P$, our result implies that this profile has a unique majority graph $G_P$. Even in the presence of Pareto-dominated houses, we can deduce preferences of all agents except for the direction of some Pareto-dominations. 
The full proof of the following result is deferred to \Cref{app:thm1}.

\begin{theorem}\label{thm:C1_Char}
	Two profiles induce the same majority graph if and only if they are rotation equivalent.
\end{theorem}

\begin{proof}[Proof sketch]
    It is easy to verify that rotation equivalent profiles indeed induce the same majority graph, so we focus on the remaining implication.
    Let $G_{P^*}$ be a majority graph that is induced by some profile $P^*$. Our goal is to find all profiles $P$ such that $G_P= G_{P^*}$. 
    As a first step, we consider a pair of houses $p,q$. We iterate over pairs of agents $x,y$ and instantiate an assignment $\mu$ in which $x$ obtains $p$ and $y$ obtains $q$. We compare $\mu$ to the assignment $\lambda$ in which $x$ and $y$ swap houses. If the two assignments create a majority tie, then $x$ and $y$ have identical preferences over $p,q$. However, if, e.g., $\mu$ is strictly majority-preferred to $\lambda$, then this means that $p \pref_x q$ and  $q \pref_y p$ for any profile $P$ with $G_P = G_{P^*}$.
    In other words, for each pair of distinct houses $p$ and $q$, we can determine whether one Pareto-dominates the other (without knowing which one) for all profiles $P$ with $G_P = G_{P^*}$. If this is not the case, then we can determine for each agent whether she ranks $p$ over $q$ in all such profiles $P$ or vice versa.

    Next, we instantiate a graph with the houses being the nodes. An edge between two houses $p$ and $q$ is added whenever not all agents prefer $p$ to $q$ or vice versa. Based on the insights from the previous paragraph, we can determine for each agent which of the houses they prefer more. This graph partitions the set of houses into connected components $H_1,\dots, H_k$. 
    We then show that within each component $H_i$, we can determine the relative ordering between all pairs of houses by querying appropriate majority comparisons, and that each agent ranks the houses of each component contiguously. Finally, we re-order the components and prove that $P$ can be decomposed as $(H_1,\dots,H_k)$ or a rotation $(H_{1+r},\dots,H_{k+r})$ thereof.
\end{proof}

\begin{remark}
    The proof of \Cref{thm:C1_Char} yields an efficient algorithm for reconstructing all profiles inducing a given majority graph in time polynomial in $n$.
Moreover, the majority graph uniquely determines the majority margins, which can also be deduced algorithmically.
    By contrast, to verify whether a given directed graph is the majority graph of a profile, one needs to check all ${\sim}(n!)^2$ edges of the graph.
\end{remark}

\begin{remark}
    Unless $n$ is small, only very few profiles admit a non-trivial decomposition, implying that they can be fully reconstructed from their majority graph. As a matter of fact, almost all majority graphs are induced by a \emph{single} preference profile. Calculations by \citet{BHS17a} demonstrate that more than 99\% of all majority graphs are induced by a single profile as soon as $n\geq 4$.
\end{remark}

\section{Majoritarian Assignment Rules}

The concept of majority graphs has given rise to numerous influential solution concepts in social choice theory, such as Condorcet winners, Copeland's rule, the top cycle, and the uncovered set \citep[see, e.g.,][]{Lasl97a,BBH15a}. In particular, all of these concepts are majoritarian, i.e., they can be computed solely based on the majority graph of a profile. As a consequence, the definitions of these concepts directly carry over to the assignment domain while preserving their natural appeal. Weak Condorcet winners, for example, are known as popular assignments in house allocation. An assignment $\mu$ is \emph{popular} if $\mu \mprefsim \lambda $ for all $\lambda\in M$. \Cref{ex:majority_graph} shows that popular assignments need not exist.
In the following, we investigate \emph{majoritarian} assignment rules, i.e., assignment rules that only depend on the majority graph. Formally, an assignment rule $F$ is majoritarian if $F(P)=F(P')$ for all profiles $P$ and $P'$ with $G_P=G_{P'}$. 

While our main focus is the study of established majoritarian voting rules in the context of assignment, 
\Cref{thm:C1_Char} implies that several well-known assignment concepts are actually majoritarian. Specifically, this result entails that an assignment rule is majoritarian if and only if it is invariant with respect to rotation equivalence. We use this fact to prove that Pareto-optimality, least unpopularity, and mixed popularity are majoritarian. By \emph{least unpopularity}, we denote the rule that returns all assignments minimizing the margin of their worst majority defeat. By \emph{mixed popularity}, we denote the rule that returns all assignments which are part of the support of some mixed popular matching. 
Formal definitions of these concepts can, for example, be found in the papers by \citet{McCu08a}, \citet{KMN11a}, and \citet{BrBu20a}.

\begin{corollary}\label{Cor:Everything_Majoritarian}
	\po, least unpopularity, and mixed popularity are majoritarian.
\end{corollary}
\begin{proof}
    Let $P$ and $P'$ be two rotation equivalent profiles w.r.t.\ some decomposition $(H_1,\dots, H_k)$ and a shift by $r\in \{1,\dots, k-1\}$.\smallskip
    
    \noindent\textbf{PO:} Let $\mu\in \po(P)$ be a Pareto-optimal assignment. By a characterization due to \citet{AbSo98a}, there is an order over the agents $\sigma=(x_1,\dots,x_n)$ such that $\mu$ is chosen by the serial dictatorship $SD_\order$ induced by $\sigma$, i.e., $\mu=SD_\sigma(P)$. Note that the agents choose houses from $H_1,\dots, H_k$ in this order because, for all $i,j\in \{1,\dots, k\}$ with $i<j$, it holds that every agent prefers every house $h\in H_i$ to every house $h'\in H_j$. Next, we partition the agents $x\in N$ into the sets $N_i=\{x\in N\colon \mu(x)\in H_i\}$ for $i\in \{1,\dots, k\}$. Let $\sigma'$ denote the order of agents such that \emph{(i)} for all $i,j\in \{1,\dots, k\}$ with $i<j$, all agents in $N_{i+r}$ are ranked before all agents in $N_{j+r}$ and \emph{(ii)} within each set $N_i$, the agents are ordered the same as in $\sigma$. Under the serial dictatorship $SD_{\sigma'}$ induced by this sequence, the agents from $N_{1+r}$ first get to choose their houses. Since $P'$ is achieved by rotating $P$ with a shift of $r$, these agents obtain precisely the same houses from $H_{1+r}$ as in $\mu$. Inductively, the agents in $N_{j+r}$ obtain precisely the houses from $H_{j+r}$ under this shifted picking sequence in the profile $P'$, and the obtained assignment is hence $\mu$. This proves that $\mu \in \po(P') $. Reversing the roles of $P$ and $P'$, we obtain that $\po(P)=\po(P')$.\smallskip

    \noindent\textbf{Least unpopularity and mixed popularity:} The proof of \Cref{thm:C1_Char} shows that the margins of all majority comparisons can be inferred from the majority graph. Since least popularity and mixed popularity only depend on these margins, they are majoritarian.
\end{proof}

By contrast, we show next via an example that the rule that returns all \emph{rank-maximal assignments} \citep{IKM+06a} is not majoritarian. To introduce rank-maximality, we define the rank of a house $p$ in a preference relation $\succ$ by $r(\succ,p)=1+|\{q\in H\mid q\succ p\}|$. Further, the rank vector of an assignment $\mu$ for a profile $P$ contains the ranks $r(\succ_x,\mu(x))$ of each agent $x\in N$ for her assigned house in increasing order. Then, an assignment $\mu$ is rank-maximal if its rank vector is lexicographically optimal, i.e., the assignment maximizes the number of agents who obtain their favorite house, subject to this condition it maximizes the number of agents who obtain their second-ranked house, and so on.

\begin{example}
    Consider the following two profiles $P$ and $P'$. They are rotation equivalent by the decomposition $(\{a,b,c\},\{d,e,f\})$ and thus induce the same majority graph.
    \[\arraycolsep=3pt P = \begin{array}{rcccccc}
    1\colon & \goodass a,&b,&c,&d,&e,&f\\
	2\colon & a,&\goodass c,&b,&d,&f,&e\\
	3\colon & \goodass b,&a,&c,&e,&d,&f\\
	4\colon & a,&b,&c,&\goodass d,&e,&f\\
    5\colon & a,&b,&c,&d,&\goodass e,&f\\
	6\colon & a,&b,&c,&d,&e,&\goodass f\\
    \end{array}
    \quad\quad
    P' = \begin{array}{rcccccc}
    1\colon \badass d,&e,&f, &\goodass a,&b,&c\\
	2\colon d,&\badass f,&e, &a,&\goodass c,&b\\
	3\colon \badass e,&d,&f, &\goodass b,&a,&c\\
	4\colon \goodass d,&e,&f, &\badass a,&b,&c\\
    5\colon d,&\goodass e,&f, &a,&\badass b,&c\\
	6\colon d,&e,&\goodass f, &a,&b,&\badass c\\
    \end{array}
    \]

    The assignment $\mu = (a,c,b,d,e,f)$ marked in blue is rank maximal in $P$, but not in $P'$, as the assignment $\lambda = (d,f,e,a,b,c)$ in red assigns two agents their top choice. To see that $\mu$ indeed is rank-maximal in $P$, note that any assignment can give at most two agents their top choices. Agent $3$ has to obtain $b$, and since $1,4,5,6$ all have the same preferences, we can assign $a$ to agent $1$. Among the remaining agents, only $2$ can still obtain her second-favorite house, $c$. Among $4,5,6$, it then does not matter how we assign $d,e,f$ for rank-maximality.
\end{example}

We conclude this section by proving another surprising relationship between two majoritarian assignment rules. Specifically, we show that all Pareto-optimal assignments are semi-popular. Semi-popularity is a weakening of popularity, which requires that an assignment is majority preferred to at least half of all assignments. More formally, an assignment $\mu$ is \emph{semi-popular} in a profile~$P$ 
if $\lvert \{\lambda \in M\mid \mu \mprefsim \lambda\}\rvert \ge  \frac {\lvert M\rvert}2$ \citep{KaVa23a}.
By $SP(P)$, we denote the set of all semi-popular assignments in a profile~$P$.

\begin{restatable}{proposition}{POSubseteqSP}\label{thm:PO_Subset_SP}
    $\po\subseteq SP$.
\end{restatable}
\begin{proof}
For our proof, we first introduce permutations on assignments. Given a permutation $\pi$ on $N$, define $\pi'\colon (N\to H)\to(N\to H)$ such that for any assignment $\mu\colon N\to H$ and agent $x\in N$, we have $\pi'(\mu)(x):=\mu(\pi(x))$. Intuitively, the assignment $\mu'=\pi'(\mu)$ is obtained from $\mu$ by assigning to agent $x$ the house that is given to agent $\pi(x)$ in $\mu$. 
For the sake of simplicity, we slightly abuse notation and refer to $\pi'$ as $\pi$ too.

	Now, fix an arbitrary profile $P$ and an assignment $\mu\in\po(P)$. 
	We consider an arbitrary permutation $\pi$ and show that $\pi(\mu)\mpref\mu$ implies that $\mu\mpref\pi^{-1}(\mu)$.
	For simplicity, we name $\pi(\mu)=:\eta$ and $\pi^{-1}(\mu)=:\lambda$.
	Let $y$ denote an arbitrary agent who strictly prefers $\eta$ to $\mu$, i.e., $\eta(y)\pref_y\mu(y)$.
	Further, let $x=\mu^{-1}\left(\eta(y)\right)$ be the agent who gets $\eta(y)$ in $\mu$. Note that $x\neq y$.
	Since $\mu$ is Pareto-optimal, it cannot be that $\mu(y)\pref_x\mu(x)$, as otherwise swapping the houses of $x$ and $y$ would be a Pareto-improvement over $\mu$.
	Therefore, $\mu(x)\pref_x\mu(y)=\lambda(x)$ and $\mu\pref_x\lambda$.
	Since $y$ was chosen arbitrarily, we see that for every agent strictly preferring $\eta$ to $\mu$, we have one other strictly preferring $\mu$ to $\lambda$. Moreover, if $\mu(x)=\eta(x)$, then $\pi(x)=x$, which implies also that $\mu(x)=\lambda(x)$. Hence, if a majority of agents prefer $\eta$ to $\mu$, a majority of agents prefer $\mu$ to $\lambda$. We lastly note that every assignment $\eta$ can be obtained by permuting $\mu$, i.e., there is some permutation $\pi$ such that $\eta=\pi(\mu)$. Hence, it follows that $\lvert \{\lambda \in M\mid \mu \mprefsim \lambda\}\rvert \ge  \frac {\lvert M\rvert}2$, so $\mu$ is semi-popular.
\end{proof}

\subsection{The Top Cycle}

We next turn to the top cycle, one of the most prominent majoritarian rules in the social choice domain \citep[e.g.,][]{Ward61a,Good71a,Smit73a,Bord76a,Schw86a,BrLe21a}. The underlying idea is very natural: popular assignments do not always exist because the majority relation $\succsim$ fails to be transitive (see \Cref{ex:majority_graph}). Instead, one can consider $\succsim^*$, the \emph{transitive closure} of $\succsim$, and simply return the maximal elements according to this relation. Formally, $\tc(P)=\left\{\mu\in M\colon\forall\lambda\in M:\mu\mprefsimt\lambda\right\}$.\footnote{One can also consider $\succ^*$, the transitive closure of the \emph{strict} part of the majority relation and return its maximal elements. The resulting SCF is known as the Schwartz set or GOCHA \citep{Schw86a}}$^,$\footnote{The top cycle is unrelated to the \emph{Top Trading Cycle (TTC)} by \citet{ShSc74a}, an assignment algorithm for settings with initial endowments.} 
Or, in other words, the top cycle returns all assignments that reach every other assignment on some path in the majority graph. 

As a first step towards understanding the top cycle in the assignment domain, we prove that it always contains all Pareto-optimal assignments. This is not true in the social choice domain \citep[see, e.g.,][Theorem 10.2.3]{Lasl97a}. The proof of the following proposition as well as \Cref{thm:Structure_TC} can be found in \Cref{app:TC}.
\begin{restatable}{proposition}{POSubseteqTC}\label{thm:PO_Subset_TC}
    $\po\subseteq \tc$.
\end{restatable}

Using \Cref{thm:PO_Subset_TC} as a stepping stone, we obtain a much stronger structural result about majority graphs in assignment: when $n\geq 5$, 
the top cycle can only contain one, two, all but one, all but two, or all assignments. 

\begin{restatable}{theorem}{TC}\label{thm:Structure_TC}
     Let $P$ be any profile with $n\geq 5$ agents and houses. Then, $\lvert \tc(P) \vert \in \{1,2,n!-2,n!-1, n!\}$.
     More precisely, we have
     \begin{enumerate}[label=\textit{(\roman*)}]
         \item $\lvert \tc(P) \vert = 1$ if all agents have distinct top choices,\footnote{This case corresponds to the profiles admitting a strongly popular assignment, which strictly majority dominate every other assignment. These are known as (strict) Condorcet winners in social choice.}
         \item $\lvert \tc(P) \vert = 2$ if all but two agents have distinct top choices. Further, these two also share the same second choice, which is not top-ranked by any other agent either,
         \item $\lvert \tc(P) \vert = n! - 2$ if \textit{(i)} and \textit{(ii)} do not hold, and all but two agents have distinct bottom choices. Further, these two also share the same second-to-bottom choice, which is not last-ranked by any other agent either,
         \item $\lvert \tc(P) \vert = n! - 1$ if \textit{(i)} and \textit{(ii)} do not hold, and all agents have distinct bottom choices, and
         \item $\lvert \tc(P) \vert = n!$ if none of the above cases holds.
     \end{enumerate}
\end{restatable}
\begin{proof}[Proof sketch]
    First, the cases \textit{(i)} and \textit{(ii)} follow relatively easily from \Cref{thm:PO_Subset_TC}: we get under the corresponding assumptions that there are $1$ or $2$ Pareto-optimal assignments, and it is easy to show that these are the only ones in the top cycle. Moreover, if the corresponding assumptions are not true, we can show that there are more than $1$ (resp. $2$) Pareto-optimal assignments. 
    
    For cases \textit{(iii)} through \textit{(v)}, we first prove via a computer-aided approach that our result holds when there are $n=5$ agents and houses. Specifically, we let the computer enumerate all profiles for $n=5$ (up to symmetries) and verify that our theorem is true. Based on this insight, we then tackle the remaining cases when $n\geq 6$. To this end, we define the bottom cycle $\bc(P)$ as the set of assignments that can be reached from every other assignment via a path in the majority graph. Further, let $P_{N',H'}$ denote the restriction of our input profile $P$ to a set of agents $N'$ and set of houses $H'$ (with $|N'|=|H'|$), and, for every assignment $\mu$, let $\mu_{N',H'}$ be the assignment from $N'$ to $H'$ induced by $\mu$. Based on our five-agent case, we then prove that there is a path from $\lambda$ to $\mu$ in the majority relation if the sets $N'=\{x\in N\colon \mu(x)\neq \lambda(x)\}$ and $H'=\mu(N')$ satisfy that $|N'|=|H'|=5$, $\lambda_{N',H'}\not\in \bc(P_{N',H'})$ or $|\bc(P_{N',H'})|>2$, and $\mu_{N',H'}\not\in \tc(P_{N',H'})$ or $|\tc(P_{N',H'})|>2$.

    Next, we show that, under cases \textit{(iii)}, \textit{(iv)}, and \textit{(v)}, there is an agent $x^*$ such that, if $x^*$ get his most preferred house $p$ in an assignment $\mu$, then $\mu$ belongs to the top cycle. To prove this claim, we consider a case distinction regarding the shape of $P$. In particular, since cases \textit{(i)} and \textit{(ii)} do not apply, we know that either (1) there is a house $p$ that is top-ranked by at least three agents, (2) there are two houses $p$ and $q$ that are both top-ranked by two agents, or
       (3) there are two houses $p$ and $q$ such that two agents top-rank $p$, one of them second-ranks $q$, and a third agent top-ranks $q$.
       
    For example, to prove our claim for case (1), let $x$, $y$, $z$ denote three agents that top-rank $p$ and let $\mu$ denote an assignment derived from a serial dictatorship where $x$ picks first. By this definition, $\mu(x)=p$ and $\mu\in\tc(P)$ because $\mu$ is Pareto-optimal. Next, let $\lambda$ denote an assignment derived from $\mu$ by swapping the houses of two arbitrary agents $u,v\in N\setminus\{x\}$. We claim that $\lambda\in\tc(P)$. To see this, we choose a set $N'\supseteq \{x,y,z,u,v\}$ with $|N'|=5$ and let $H'=\lambda(N)$. In $P_{N',H'}$, the agents $x,y,z$ still top-rank $p$, so we infer from cases (i) and (ii) that $|\tc(P_{N',H'})|>2$. Further, since $\lambda(x)=\mu(x)=p$, $\lambda_{N',H'}\not\in\bc(P_{N',H'})$ unless $|\bc(P_{N',H'})|>2$. Hence, the insights of the previous paragraph show that $\lambda\mprefsimt \mu$. By repeatedly applying this construction, it follows that every assignment $\lambda$ with $\lambda(x)=p$ is in the top cycle. The other cases rely on similar ideas. 

    As the next step, we prove that $\tc(P)$ contains all assignments except for so-called Pareto-pessimal ones, which are those that do not Pareto-dominate any other assignment. To prove this claim, we focus on an assignment $\mu$ that Pareto-dominates another assignment $\lambda$. Hence, for all agents $x$ with $\mu(x)\neq\lambda(x)$, it holds that $\mu(x)\succ_x\lambda(x)$. Further, let $x^*$ denote the agent such that, if $x^*$ gets his favorite house $p$ in an assignment, the assignment is in the top cycle. We show that we can modify $\lambda$ such that agent $x^*$ is assigned his favorite house $p$ while maintaining that $\mu$ majority dominates the resulting assignment $\eta$. Since $\eta\in\tc(P)$ by our previous discussion, this proves that $\mu\in\tc(P)$, too.

    Finally, we observe that the top cycle of a profile is the bottom cycle of the profile where we reversed all agents' preference relations. This implies that all Pareto-pessimal assignments are contained in the bottom cycle and that the bottom cycle contains all assignments that are not Pareto-optimal if $|\bc(P)|>2$. For cases \emph{(iii)} and \emph{(iv)}, our theorem then follows by showing that there are precisely two or one Pareto-pessimal assignments. By contrast, for case \textit{(v)}, we show that either all assignments are Pareto-optimal and thus in TC, or there is an assignment that is not Pareto-optimal but in the top cycle. Since this assignment is also in the bottom cycle in this case, $\tc(P)=\bc(P)$, which implies that $\tc(P)=M$. 
\end{proof}

\Cref{thm:Structure_TC} shows that deterministic assignments are highly unstable with respect to majority deviations. Indeed, given a starting assignment and a target assignment, one can almost always convince the agents to transition from one to the other by presenting intermediate assignments that are weakly majority preferred.
To illustrate this point, consider the following poor assignment that reaches every other assignment via some majority path.

\begin{example}\label{ex:bad_in_tc}
	In the subsequent profile $P$, the assignment $\mu$ marked in red is obviously not desirable. It fails to be Pareto-optimal, and many agents even receive their least-preferred house. Nevertheless, it is contained in the top cycle. A path of dominations via which $\mu$ reaches a serial dictatorship (and by virtue of \Cref{thm:PO_Subset_TC} the entire top cycle) is given in \Cref{app:path_bad}. The rough idea is to repeatedly reassign to some agents slightly worse houses while improving other agents' assignments significantly. For example, we can make agents $4,5$ worse by assigning $c,a$ to them, respectively. However, this frees houses $d,e$ which we assign to agents $2$ and $7$, thereby making these agents significantly happier.
    
    \[\arraycolsep=3pt P = \begin{array}{rccccccc}
	1\colon & f,&b,&d,&e,&c,&a,&\badass g\\
	2\colon & d,&f,&g,&a,&b,&e,&\badass c\\
	3\colon & d,&a,&c,&g,&e,&b,&\badass f\\
	4\colon & a,&d,&b,&f,&g,&\badass e,&c\\
	5\colon & c,&g,&e,&b,&f,&\badass d,&a\\
	6\colon & f,&a,&e,&d,&g,&c,&\badass b\\
	7\colon & c,&d,&e,&b,&g,&f,&\badass a\\
    \end{array}
    \]
\end{example}

\begin{remark}{
    For completeness, we also consider the cases $n<5$ for \Cref{thm:Structure_TC}. Clearly, there exists only one assignment for $n=1$, and two assignments for $n=2$.
    For $n=3$, we have found via a computer-aided approach that the top cycle has size either $1$, $2$, $n!-2=4$, or $n!=6$. However, in contrast to case \textit{(iv)} of \Cref{thm:Structure_TC}, the top cycle may contain $n!-2$ assignments even though all agents have pairwise distinct bottom choices. This happens, for example in the following profile, where $\mu= (b,c,a)$ (in red) and $\lambda=(c,b,a)$ do not belong to the top cycle.
    \[\arraycolsep=3pt P = \begin{array}{rccc}
    1\colon & a, &\badass b, & c\\
	2\colon & a, &\badass c, & b\\
	3\colon & c, &b,&  \badass a\\
	\end{array}
    \]
    Lastly, for $n=4$, we found via our computer-aided approach that the top cycle can, in addition to the five sizes described in \Cref{thm:Structure_TC}, also have a size of $n!-3$. Up to symmetries, this happens precisely in the following two profiles $P$ and $P'$, where the assignments $\mu= (c,d,b,a)$ (marked in red), $\lambda = (d,c,a,b)$, and $\eta=(d,c,b,a)$ are respectively not contained in the top cycle.
    \[\arraycolsep=3pt P = \begin{array}{rcccc}
    1\colon & a, &b, & \badass c, & d\\
	2\colon & a, &b , & \badass d, & c\\
	3\colon & c, &d, & a, & \badass b\\
	4\colon & c, &d, & b, & \badass a\\
    \end{array}
    \quad
    P' = \begin{array}{rcccc}
    1\colon & a, &b, & \badass c, & d\\
	2\colon & a, &b , & \badass d, & c\\
	3\colon & d, &c, & a, & \badass b\\
	4\colon & d, &c, & b, & \badass a\\
    \end{array}
    \]
}
\end{remark}

\begin{remark}
    \Cref{thm:Structure_TC} stands in stark contrast to classic social choice, where the top cycle may have any number of elements, even when there are at most three agents.
Further, in social choice theory, \tc can be computed in linear time in the size of the profile \citep{BFH09b,BBH15a}. \Cref{thm:Structure_TC} implies that in assignment, computing and returning a concise representation of the (possibly exponentially large) top cycle is possible in \emph{sub-linear} time. 
\end{remark} 

\subsection{Uncovered Sets}

As our final contribution, we turn to another technique addressing the non-transitivity of the majority relation: uncovered sets. These sets are based on covering relations, which are natural \emph{transitive subrelations} of the majority relation. Just as in the definition of the top cycle, we can take the maximal elements for each of these relations, defining an uncovered set that refines the top cycle. Just as the top cycle, uncovered sets have been extensively studied in social choice theory \citep[see, e.g.,][]{Bord83a,BLS92a,BrFi08b,Dugg11a}. 

The presence of majority ties in the assignment domain allows for multiple definitions of covering relations and uncovered sets, and we will subsequently define the three most common ones. Given a profile $P$, an assignment $\mu$ \emph{Bordes covers} another assignment $\lambda$, if $\mu \pref \lambda$, and for every third assignment $\eta$, we have that $\lambda \pref \eta$ implies $\mu \pref \eta$. Similarly, $\mu$ \emph{Gillies covers} $\lambda$, if $\mu \pref \lambda$, and for every $\eta$, we have that $\eta \pref \mu$ implies $\eta \pref \lambda$. Finally, $\mu$ \emph{McKelvey covers} $\lambda$ if it Bordes and Gillies covers it. 
Each of the three covering relations gives rise to a corresponding \emph{uncovered set} (\uc). It returns the maximal assignments of the covering relation, i.e., $\uc(P) = \{\mu\in M \colon \text{no $\lambda\in M$ covers $\mu$}\}$.
Whenever we refer to covering or \uc without 
further specification, we mean McKelvey covering. 
All three uncovered sets can be characterized as assignments that reach all other assignments via some majority path of length at most~2. For Bordes, the first segment of any path of length~2 must be strict; for Gillies, the second segment must be strict; and for McKelvey, one of the two segments must be strict.
This immediately implies that all uncovered sets are contained in the top cycle.
Moreover, both the Bordes and the Gillies uncovered set are refinements of the McKelvey uncovered set.

From the general social choice setting, we know that $\uc\subseteq \po$ \citep{Fish77a}. This inherits to assignment, and we can easily prove that the inclusion is strict on this domain, too.

\begin{example}\label{ex:UC_subseteq_PO}
	In the following profile, the assignment $\mu= (c,a,b)$ in blue McKelvey-covers the assignment in red $\lambda= (a,b,c)$, even though $\lambda$ is Pareto-optimal.

    \[\arraycolsep=3pt P = \begin{array}{rccc}
    1\colon & \badass a,&\goodass c,&b\\
	2\colon & \goodass a,&\badass b,&c\\
	3\colon & \goodass b,&a,&\badass c\\
    \end{array}
    \]
    Recall that all serial dictatorships are Pareto-optimal. Hence, this example illustrates that \po fails to distinguish between ``good'' picking sequences and ``bad'' ones in which the agents take away each other's favorite houses in unfortunate ways. 
    This effect occurs for arbitrarily large numbers of agents. Thus, the set of uncovered assignments can be seen as a particularly attractive subset of the set of Pareto-optimal assignments.
\end{example}

To see how decisive \uc is, we computed its choice sets and tracked the occurring cardinalities while iterating over all preference profiles up to symmetry for $n=5$ exhaustively. The resulting graph is depicted in \Cref{fig:uc5}. The corresponding code can be found on Zenodo \citep{Schl26a}. Further, we sampled profiles for $n=7$ agents drawing each agent's preferences uniformly at random,
depicted in \Cref{fig:uc7}.
It turns out that the Bordes-\uc is almost indistinguishable from the McKelvey-\uc, while the Gillies-\uc is, on average, the most discriminating one. This can be explained as follows: Under the Gillies-\uc, for an assignment $\lambda$ to not be covered despite $\mu \pref \lambda$, there needs to exist another assignment $\eta$ such that $\lambda \mprefsim \eta \pref \mu$. However, for small numbers of agents, there are many profiles admitting popular assignments. If $\mu$ is such a popular assignment, there exists no $\eta$ with $\eta\pref \mu$, and hence $\mu$ automatically Gillies-covers all $\lambda$ with $\mu \pref \lambda$.
Most notably, many profiles in both simulations admitted an uncovered set of size two. This finding suggests that \uc is much more discriminative in assignment than in general social choice. However, can this already be explained by \po being more discriminative in assignment than in social choice?

To investigate how much \uc differs from \po, we exhaustively studied the case when $n=5$. For this, we utilize that the rules are symmetric with respect to permuting agents and houses. We therefore fix the preferences of agent $1$, and further demand that the preferences of agents $2$ through $5$ are ordered lexicographically. This results in roughly $9$ million profiles to be checked, which can be done within a few days on a computer.
In \Cref{fig:ucPOdiffquot}, we depict the percentage of profiles $P$ for which the ratio $\frac{\lvert \uc(P)\rvert}{\lvert \po(P)\rvert}$ is at most $x\in [0,1]$, as a function of $x$. The results suggest that, indeed, \uc is significantly more discriminative than \po. Hence, this consolidates that \uc is an interesting refinement of \po, and it seems worthwhile to further investigate the properties of \uc.

\begin{figure}
	\begin{tikzpicture}[xscale=1.4,yscale=0.9]\datavisualization [scientific axes=clean,
		x axis={label={Cardinality}},
		y axis={label={Frequency (as percentage)}},
legend=north east inside,
		visualize as line/.list={mckelvey, bordes, gillies, serialdictator},
		legend=north east inside,
		style sheet=strong colors,
mckelvey={label in legend={text=McKelvey,
            label in legend line coordinates={(-0.6em,0),(0,0)}}},
  bordes={label in legend={text=Bordes,
            label in legend line coordinates={(-0.6em,0),(0,0)}}},
  gillies={label in legend={text=Gillies,
            label in legend line coordinates={(-0.6em,0),(0,0)}}},
  serialdictator={label in legend={text=Pareto,
            label in legend line coordinates={(-0.6em,0),(0,0)}}}
        ]
		data [set=mckelvey] {x, y
1, 3.6544721
2, 39.8530175
3, 3.4260676
4, 2.3621626
5, 0.5392884
6, 4.2550914
7, 4.3504361
8, 6.5946073
9, 3.4030465
10, 4.0027846
11, 2.0853367
12, 3.6763587
13, 2.3160433
14, 2.9811106
15, 1.8014833
16, 2.2595149
17, 1.1513301
18, 1.7352178
19, 0.9706310
20, 1.4752116
21, 0.6816557
22, 1.0506321
23, 0.4975420
24, 0.9130012
25, 0.3374959
26, 0.5826320
27, 0.2844592
28, 0.5563615
29, 0.2056478
30, 0.3774468
31, 0.1378512
32, 0.3156644
33, 0.0902669
34, 0.1978052
35, 0.0753418
36, 0.2024314
37, 0.0481901
38, 0.1131999
39, 0.0448306
40, 0.0819287
41, 0.0099685
42, 0.0867972
43, 0.0213138
44, 0.0336615
45, 0.0083713
46, 0.0467582
47, 0.0033595
48, 0.0240675
49, 0.0003304
50, 0.0242217
51, 0.0000000
52, 0.0117749
53, 0.0027537
54, 0.0144405
55, 0.0000000
56, 0.0037340
57, 0.0000000
58, 0.0026876
59, 0.0000000
60, 0.0094728
61, 0.0000000
62, 0.0001322
63, 0.0000000
64, 0.0034477
65, 0.0000000
66, 0.0014099
67, 0.0000000
68, 0.0007931
69, 0.0000000
70, 0.0000000
71, 0.0000000
72, 0.0012337
73, 0.0000000
74, 0.0000000
75, 0.0000000
76, 0.0000000
77, 0.0000000
78, 0.0013548
79, 0.0000000
80, 0.0000000
81, 0.0000000
82, 0.0000000
83, 0.0000000
84, 0.0000881
85, 0.0000000
86, 0.0000000
87, 0.0000000
88, 0.0000000
89, 0.0000000
90, 0.0000000
91, 0.0000000
92, 0.0000000
93, 0.0000000
94, 0.0000000
95, 0.0000000
96, 0.0002203
97, 0.0000000
98, 0.0000000
99, 0.0000000
100, 0.0000000
101, 0.0000000
102, 0.0000000
103, 0.0000000
104, 0.0000000
105, 0.0000000
106, 0.0000000
107, 0.0000000
108, 0.0000000
109, 0.0000000
110, 0.0000000
111, 0.0000000
112, 0.0000000
113, 0.0000000
114, 0.0000000
115, 0.0000000
116, 0.0000000
117, 0.0000000
118, 0.0000000
119, 0.0000000
120, 0.0000110

}
		data [set=bordes] {x, y
1, 3.6544721
2, 39.8530175
3, 3.4260676
4, 2.3621626
5, 0.5392884
6, 4.3100776
7, 4.4350304
8, 6.8968336
9, 3.5037225
10, 3.8245859
11, 1.9476507
12, 3.6884971
13, 2.3464443
14, 2.9472839
15, 1.8159128
16, 2.2320879
17, 1.1115113
18, 1.7317591
19, 0.9290499
20, 1.4752116
21, 0.6729540
22, 1.0433292
23, 0.4777152
24, 0.8953884
25, 0.3326493
26, 0.5729609
27, 0.2901319
28, 0.5418218
29, 0.2008563
30, 0.3553730
31, 0.1359236
32, 0.3160609
33, 0.0896060
34, 0.1859862
35, 0.0753968
36, 0.1949193
37, 0.0500626
38, 0.1098183
39, 0.0424624
40, 0.0793402
41, 0.0098583
42, 0.0848366
43, 0.0198819
44, 0.0340690
45, 0.0087568
46, 0.0448636
47, 0.0031943
48, 0.0236600
49, 0.0003304
50, 0.0240895
51, 0.0000000
52, 0.0113894
53, 0.0030842
54, 0.0140770
55, 0.0000000
56, 0.0038662
57, 0.0000000
58, 0.0025555
59, 0.0000000
60, 0.0093737
61, 0.0000000
62, 0.0001322
63, 0.0000000
64, 0.0034477
65, 0.0000000
66, 0.0014099
67, 0.0000000
68, 0.0007931
69, 0.0000000
70, 0.0000000
71, 0.0000000
72, 0.0012337
73, 0.0000000
74, 0.0000000
75, 0.0000000
76, 0.0000000
77, 0.0000000
78, 0.0013548
79, 0.0000000
80, 0.0000000
81, 0.0000000
82, 0.0000000
83, 0.0000000
84, 0.0000881
85, 0.0000000
86, 0.0000000
87, 0.0000000
88, 0.0000000
89, 0.0000000
90, 0.0000000
91, 0.0000000
92, 0.0000000
93, 0.0000000
94, 0.0000000
95, 0.0000000
96, 0.0002203
97, 0.0000000
98, 0.0000000
99, 0.0000000
100, 0.0000000
101, 0.0000000
102, 0.0000000
103, 0.0000000
104, 0.0000000
105, 0.0000000
106, 0.0000000
107, 0.0000000
108, 0.0000000
109, 0.0000000
110, 0.0000000
111, 0.0000000
112, 0.0000000
113, 0.0000000
114, 0.0000000
115, 0.0000000
116, 0.0000000
117, 0.0000000
118, 0.0000000
119, 0.0000000
120, 0.0000110

}
		data [set=gillies] {x, y
1, 3.6544721
2, 53.2652614
3, 0.0000000
4, 23.9052809
5, 2.9502249
6, 3.8656493
7, 0.0356882
8, 3.9740467
9, 1.3059239
10, 1.0422608
11, 0.2564264
12, 1.0332837
13, 0.4332151
14, 0.6584914
15, 0.1570171
16, 0.4232026
17, 0.0495119
18, 0.3778654
19, 0.1543735
20, 0.3368460
21, 0.0731939
22, 0.2510181
23, 0.0509438
24, 0.2128625
25, 0.0478046
26, 0.1567527
27, 0.0864668
28, 0.1767227
29, 0.0686227
30, 0.1125060
31, 0.0475843
32, 0.1350975
33, 0.0422971
34, 0.0902449
35, 0.0319982
36, 0.1169229
37, 0.0268763
38, 0.0655826
39, 0.0278676
40, 0.0504261
41, 0.0064988
42, 0.0642828
43, 0.0197166
44, 0.0207961
45, 0.0101337
46, 0.0305002
47, 0.0027537
48, 0.0221840
49, 0.0003304
50, 0.0211486
51, 0.0001652
52, 0.0118079
53, 0.0027537
54, 0.0116097
55, 0.0004406
56, 0.0033375
57, 0.0000000
58, 0.0026876
59, 0.0000000
60, 0.0094618
61, 0.0000000
62, 0.0000000
63, 0.0000000
64, 0.0034477
65, 0.0000000
66, 0.0014099
67, 0.0000000
68, 0.0007931
69, 0.0000000
70, 0.0000000
71, 0.0000000
72, 0.0012337
73, 0.0000000
74, 0.0000000
75, 0.0000000
76, 0.0000000
77, 0.0000000
78, 0.0013548
79, 0.0000000
80, 0.0000000
81, 0.0000000
82, 0.0000000
83, 0.0000000
84, 0.0000881
85, 0.0000000
86, 0.0000000
87, 0.0000000
88, 0.0000000
89, 0.0000000
90, 0.0000000
91, 0.0000000
92, 0.0000000
93, 0.0000000
94, 0.0000000
95, 0.0000000
96, 0.0002203
97, 0.0000000
98, 0.0000000
99, 0.0000000
100, 0.0000000
101, 0.0000000
102, 0.0000000
103, 0.0000000
104, 0.0000000
105, 0.0000000
106, 0.0000000
107, 0.0000000
108, 0.0000000
109, 0.0000000
110, 0.0000000
111, 0.0000000
112, 0.0000000
113, 0.0000000
114, 0.0000000
115, 0.0000000
116, 0.0000000
117, 0.0000000
118, 0.0000000
119, 0.0000000
120, 0.0000110

}
		data [set=serialdictator] {x, y
1, 3.6544721
2, 2.5124496
3, 1.5226967
4, 2.5273857
5, 4.5680901
6, 5.9751526
7, 5.3171018
8, 6.9472817
9, 5.3369286
10, 7.0444549
11, 5.6466670
12, 6.0978033
13, 4.3392010
14, 5.7113463
15, 3.9414537
16, 4.1690431
17, 2.8418385
18, 3.7217069
19, 2.3173100
20, 2.9573295
21, 1.7685488
22, 2.1265764
23, 1.1708264
24, 1.6088771
25, 0.8298058
26, 1.0405535
27, 0.4796979
28, 0.8608347
29, 0.3942776
30, 0.4989850
31, 0.2254745
32, 0.4259674
33, 0.1764033
34, 0.2416995
35, 0.0916438
36, 0.2370181
37, 0.0555150
38, 0.1398449
39, 0.0568918
40, 0.0911812
41, 0.0064988
42, 0.1001583
43, 0.0209283
44, 0.0348511
45, 0.0110149
46, 0.0495229
47, 0.0035248
48, 0.0235829
49, 0.0000000
50, 0.0256977
51, 0.0003304
52, 0.0116097
53, 0.0029189
54, 0.0144405
55, 0.0000000
56, 0.0037340
57, 0.0000000
58, 0.0026876
59, 0.0000000
60, 0.0094728
61, 0.0000000
62, 0.0001322
63, 0.0000000
64, 0.0034477
65, 0.0000000
66, 0.0014099
67, 0.0000000
68, 0.0007931
69, 0.0000000
70, 0.0000000
71, 0.0000000
72, 0.0012337
73, 0.0000000
74, 0.0000000
75, 0.0000000
76, 0.0000000
77, 0.0000000
78, 0.0013548
79, 0.0000000
80, 0.0000000
81, 0.0000000
82, 0.0000000
83, 0.0000000
84, 0.0000881
85, 0.0000000
86, 0.0000000
87, 0.0000000
88, 0.0000000
89, 0.0000000
90, 0.0000000
91, 0.0000000
92, 0.0000000
93, 0.0000000
94, 0.0000000
95, 0.0000000
96, 0.0002203
97, 0.0000000
98, 0.0000000
99, 0.0000000
100, 0.0000000
101, 0.0000000
102, 0.0000000
103, 0.0000000
104, 0.0000000
105, 0.0000000
106, 0.0000000
107, 0.0000000
108, 0.0000000
109, 0.0000000
110, 0.0000000
111, 0.0000000
112, 0.0000000
113, 0.0000000
114, 0.0000000
115, 0.0000000
116, 0.0000000
117, 0.0000000
118, 0.0000000
119, 0.0000000
120, 0.0000110

};
		\begin{scope}[xshift=35,yshift=20,xscale=0.5,yscale=0.7]
\clip (-1,-1) -- (-1,3.2) -- (0,3.2) -- (0,3.1) -- (6,3.1) -- (6,-1) -- (-1,-1);
\datavisualization [scientific axes=clean,
x axis={max value=25},
y axis={max value=10},
visualize as line/.list={mckelvey, bordes, gillies, serialdictator},
		style sheet=strong colors
]
		data [set=mckelvey] {x, y
			1, 3.6544721
			2, 39.8530175
			3, 3.4260676
			4, 2.3621626
			5, 0.5392884
			6, 4.2550914
			7, 4.3504361
			8, 6.5946073
			9, 3.4030465
			10, 4.0027846
			11, 2.0853367
			12, 3.6763587
			13, 2.3160433
			14, 2.9811106
			15, 1.8014833
			16, 2.2595149
			17, 1.1513301
			18, 1.7352178
			19, 0.9706310
			20, 1.4752116
			21, 0.6816557
			22, 1.0506321
			23, 0.4975420
			24, 0.9130012
			25, 0.3374959
		}
		data [set=bordes] {x, y
			1, 3.6544721
			2, 39.8530175
			3, 3.4260676
			4, 2.3621626
			5, 0.5392884
			6, 4.3100776
			7, 4.4350304
			8, 6.8968336
			9, 3.5037225
			10, 3.8245859
			11, 1.9476507
			12, 3.6884971
			13, 2.3464443
			14, 2.9472839
			15, 1.8159128
			16, 2.2320879
			17, 1.1115113
			18, 1.7317591
			19, 0.9290499
			20, 1.4752116
			21, 0.6729540
			22, 1.0433292
			23, 0.4777152
			24, 0.8953884
			25, 0.3326493
		}
		data [set=gillies] {x, y
			1, 3.6544721
			2, 53.2652614
			3, 0.0000000
			4, 23.9052809
			5, 2.9502249
			6, 3.8656493
			7, 0.0356882
			8, 3.9740467
			9, 1.3059239
			10, 1.0422608
			11, 0.2564264
			12, 1.0332837
			13, 0.4332151
			14, 0.6584914
			15, 0.1570171
			16, 0.4232026
			17, 0.0495119
			18, 0.3778654
			19, 0.1543735
			20, 0.3368460
			21, 0.0731939
			22, 0.2510181
			23, 0.0509438
			24, 0.2128625
			25, 0.0478046
		}
		data [set=serialdictator] {x, y
			1, 3.6544721
			2, 2.5124496
			3, 1.5226967
			4, 2.5273857
			5, 4.5680901
			6, 5.9751526
			7, 5.3171018
			8, 6.9472817
			9, 5.3369286
			10, 7.0444549
			11, 5.6466670
			12, 6.0978033
			13, 4.3392010
			14, 5.7113463
			15, 3.9414537
			16, 4.1690431
			17, 2.8418385
			18, 3.7217069
			19, 2.3173100
			20, 2.9573295
			21, 1.7685488
			22, 2.1265764
			23, 1.1708264
			24, 1.6088771
			25, 0.8298058
		};
\draw[->] (0,0) -- (-0.7,-0.6);
	\end{scope}
	\end{tikzpicture}
	\caption{Size distributions of UCs for $n = 5$. The high peak is at size $2$ for all of them. In total, there are $9,078,630$ profiles (up to symmetries) and there are $5!=120$ different assignments.
}\label{fig:uc5}
    \Description{A plot showing data for the four assignment rules McKelvey uncovered set, Bordes uncovered set, Gillies uncovered set, and the Pareto rule. It depicts the frequency with which each cardinality occurs across all profiles for five agents and houses. It can be seen that a size of two occurs way more often for the uncovered set than for the Pareto rule. In general, the uncovered sets often return fewer outcomes than the Pareto rule. This is particularly true for the Gillies uncovered set.}
\end{figure}
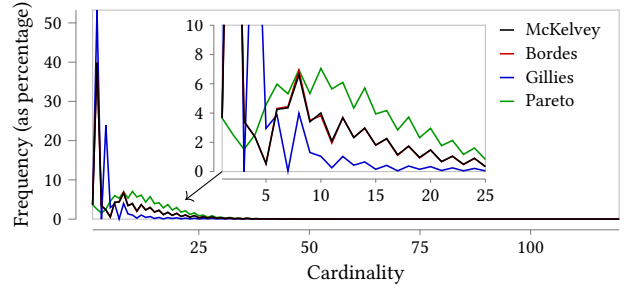
 
\begin{figure}
	\begin{tikzpicture}[xscale=1.4,yscale=0.9]\datavisualization [scientific axes=clean,
		x axis={label={Cardinality}},
		y axis={label={Frequency (as percentage)}},
legend=north east inside,
		visualize as line/.list={mckelvey, bordes, gillies, serialdictator},
		legend=north east inside,
		style sheet=strong colors,
mckelvey={label in legend={text=McKelvey,
            label in legend line coordinates={(-0.6em,0),(0,0)}}},
  bordes={label in legend={text=Bordes,
            label in legend line coordinates={(-0.6em,0),(0,0)}}},
  gillies={label in legend={text=Gillies,
            label in legend line coordinates={(-0.6em,0),(0,0)}}},
  serialdictator={label in legend={text=Pareto,
            label in legend line coordinates={(-0.6em,0),(0,0)}}}
        ]
		data [set=mckelvey] {x, y
1, 0.7000000
2, 11.5000000
3, 0.9000000
4, 1.3000000
5, 0.6000000
6, 1.4000000
7, 1.4000000
8, 2.0000000
9, 2.6000000
10, 2.1000000
11, 1.5000000
12, 1.7000000
13, 1.8000000
14, 2.1000000
15, 2.0000000
16, 1.4000000
17, 1.6000000
18, 1.9000000
19, 1.8000000
20, 1.2000000
21, 0.9000000
22, 0.8000000
23, 1.0000000
24, 1.2000000
25, 1.2000000
26, 1.4000000
27, 0.9000000
28, 1.7000000
29, 1.1000000
30, 0.8000000
31, 0.7000000
32, 1.1000000
33, 0.9000000
34, 1.0000000
35, 1.1000000
36, 0.7000000
37, 0.5000000
38, 0.7000000
39, 1.4000000
40, 0.5000000
41, 0.6000000
42, 1.2000000
43, 0.8000000
44, 0.5000000
45, 0.9000000
46, 0.5000000
47, 0.4000000
48, 1.0000000
49, 0.8000000
50, 0.5000000
51, 0.8000000
52, 0.5000000
53, 1.1000000
54, 0.2000000
55, 0.3000000
56, 0.6000000
57, 0.3000000
58, 0.5000000
59, 0.8000000
60, 0.6000000
61, 0.3000000
62, 0.4000000
63, 0.5000000
64, 0.5000000
65, 0.4000000
66, 0.4000000
67, 0.5000000
68, 0.4000000
69, 0.6000000
70, 0.0000000
71, 0.3000000
72, 0.6000000
73, 0.7000000
74, 0.7000000
75, 0.2000000
76, 0.4000000
77, 0.3000000
78, 0.4000000
79, 0.6000000
80, 0.8000000
81, 0.3000000
82, 0.2000000
83, 0.6000000
84, 0.2000000
85, 0.1000000
86, 0.4000000
87, 0.3000000
88, 0.1000000
89, 0.1000000
90, 0.2000000
91, 0.3000000
92, 0.2000000
93, 0.4000000
94, 0.4000000
95, 0.2000000
96, 0.3000000
97, 0.1000000
98, 0.2000000
99, 0.2000000
100, 0.4000000
101, 0.4000000
102, 0.2000000
103, 0.2000000
104, 0.1000000
105, 0.2000000
106, 0.1000000
107, 0.4000000
108, 0.2000000
109, 0.0000000
110, 0.2000000
111, 0.2000000
112, 0.3000000
113, 0.0000000
114, 0.4000000
115, 0.3000000
116, 0.2000000
117, 0.5000000
118, 0.2000000
119, 0.0000000
120, 0.0000000
121, 0.3000000
122, 0.1000000
123, 0.2000000
124, 0.1000000
125, 0.1000000
126, 0.1000000
127, 0.0000000
128, 0.1000000
129, 0.0000000
130, 0.2000000
131, 0.0000000
132, 0.1000000
133, 0.1000000
134, 0.1000000
135, 0.5000000
136, 0.1000000
137, 0.0000000
138, 0.4000000
139, 0.0000000
140, 0.1000000
141, 0.0000000
142, 0.1000000
143, 0.1000000
144, 0.1000000
145, 0.1000000
146, 0.0000000
147, 0.0000000
148, 0.0000000
149, 0.0000000
150, 0.0000000
151, 0.2000000
152, 0.0000000
153, 0.1000000
154, 0.3000000
155, 0.1000000
156, 0.0000000
157, 0.0000000
158, 0.0000000
159, 0.0000000
160, 0.0000000
161, 0.1000000
162, 0.1000000
163, 0.3000000
164, 0.0000000
165, 0.2000000
166, 0.0000000
167, 0.0000000
168, 0.0000000
169, 0.0000000
170, 0.0000000
171, 0.0000000
172, 0.3000000
173, 0.0000000
174, 0.0000000
175, 0.1000000
176, 0.0000000
177, 0.1000000
178, 0.0000000
179, 0.1000000
180, 0.0000000
181, 0.0000000
182, 0.2000000
183, 0.1000000
184, 0.1000000
185, 0.0000000
186, 0.0000000
187, 0.0000000
188, 0.1000000
189, 0.2000000
190, 0.0000000
191, 0.0000000
192, 0.0000000
193, 0.0000000
194, 0.0000000
195, 0.0000000
196, 0.1000000
197, 0.1000000
198, 0.0000000
199, 0.0000000
200, 0.0000000
201, 0.1000000
202, 0.0000000
203, 0.0000000
204, 0.0000000
205, 0.1000000
206, 0.0000000
207, 0.1000000
208, 0.0000000
209, 0.2000000
210, 0.1000000
211, 0.0000000
212, 0.1000000
213, 0.0000000
214, 0.0000000
215, 0.0000000
216, 0.1000000
217, 0.0000000
218, 0.0000000
219, 0.0000000
220, 0.0000000
221, 0.0000000
222, 0.1000000
223, 0.1000000
224, 0.0000000
225, 0.0000000
226, 0.0000000
227, 0.0000000
228, 0.0000000
229, 0.0000000
230, 0.0000000
231, 0.0000000
232, 0.1000000
233, 0.0000000
234, 0.1000000
235, 0.0000000
236, 0.0000000
237, 0.0000000
238, 0.0000000
239, 0.0000000
240, 0.0000000
241, 0.0000000
242, 0.0000000
243, 0.0000000
244, 0.1000000
245, 0.0000000
246, 0.0000000
247, 0.0000000
248, 0.0000000
249, 0.0000000
250, 0.0000000
251, 0.0000000
252, 0.0000000
253, 0.0000000
254, 0.0000000
255, 0.0000000
256, 0.0000000
257, 0.0000000
258, 0.0000000
259, 0.0000000
260, 0.1000000
261, 0.0000000
262, 0.0000000
263, 0.0000000
264, 0.0000000
265, 0.0000000
266, 0.1000000
267, 0.0000000
268, 0.0000000
269, 0.0000000
270, 0.0000000
271, 0.0000000
272, 0.0000000
273, 0.0000000
274, 0.0000000
275, 0.0000000
276, 0.0000000
277, 0.0000000
278, 0.0000000
279, 0.0000000
280, 0.0000000
281, 0.0000000
282, 0.0000000
283, 0.0000000
284, 0.0000000
285, 0.0000000
286, 0.0000000
287, 0.0000000
288, 0.0000000
289, 0.0000000
290, 0.0000000
291, 0.0000000
292, 0.0000000
293, 0.0000000
294, 0.0000000
295, 0.0000000
296, 0.0000000
297, 0.0000000
298, 0.0000000
299, 0.0000000
300, 0.0000000
301, 0.0000000
302, 0.0000000
303, 0.0000000
304, 0.0000000
305, 0.0000000
306, 0.0000000
307, 0.0000000
308, 0.1000000
309, 0.0000000
310, 0.0000000
311, 0.0000000
312, 0.0000000
313, 0.0000000
314, 0.0000000
315, 0.0000000
316, 0.0000000
317, 0.0000000
318, 0.0000000
319, 0.0000000
320, 0.0000000
321, 0.0000000
322, 0.0000000
323, 0.0000000
324, 0.0000000
325, 0.0000000
326, 0.0000000
327, 0.0000000
328, 0.0000000
329, 0.0000000
330, 0.0000000
331, 0.0000000
332, 0.0000000
333, 0.0000000
334, 0.0000000
335, 0.0000000
336, 0.0000000
337, 0.0000000
338, 0.0000000
339, 0.0000000
340, 0.0000000
341, 0.0000000
342, 0.0000000
343, 0.0000000
344, 0.0000000
345, 0.0000000
346, 0.0000000
347, 0.0000000
348, 0.0000000
349, 0.0000000
350, 0.0000000
351, 0.0000000
352, 0.0000000
353, 0.0000000
354, 0.0000000
355, 0.0000000
356, 0.0000000
357, 0.0000000
358, 0.0000000
359, 0.0000000
360, 0.0000000
361, 0.0000000
362, 0.0000000
363, 0.0000000
364, 0.0000000
365, 0.0000000
366, 0.0000000
367, 0.0000000
368, 0.0000000
369, 0.0000000
370, 0.0000000
371, 0.0000000
372, 0.0000000
373, 0.0000000
374, 0.0000000
375, 0.0000000
376, 0.0000000
377, 0.0000000
378, 0.0000000
379, 0.0000000
380, 0.0000000
381, 0.0000000
382, 0.0000000
383, 0.0000000
384, 0.0000000
385, 0.0000000
386, 0.0000000
387, 0.0000000
388, 0.0000000
389, 0.0000000
390, 0.0000000
391, 0.0000000
392, 0.0000000
393, 0.0000000
394, 0.0000000
395, 0.1000000
396, 0.0000000
397, 0.0000000
398, 0.0000000
399, 0.0000000
400, 0.0000000
401, 0.0000000
402, 0.0000000
403, 0.0000000
404, 0.0000000
405, 0.0000000
406, 0.0000000
407, 0.0000000
408, 0.0000000
409, 0.0000000
410, 0.0000000
411, 0.0000000
412, 0.0000000
413, 0.0000000
414, 0.0000000
415, 0.0000000
416, 0.0000000
417, 0.0000000
418, 0.0000000
419, 0.0000000
420, 0.0000000
421, 0.0000000
422, 0.0000000
423, 0.0000000
424, 0.0000000
425, 0.0000000
426, 0.0000000
427, 0.0000000
428, 0.0000000
429, 0.0000000
430, 0.0000000
431, 0.0000000
432, 0.0000000
433, 0.0000000
434, 0.0000000
435, 0.0000000
436, 0.0000000
437, 0.0000000
438, 0.0000000
439, 0.0000000
440, 0.0000000
441, 0.0000000
442, 0.0000000
443, 0.0000000
444, 0.0000000
445, 0.0000000
446, 0.0000000
447, 0.0000000
448, 0.0000000
449, 0.0000000
450, 0.0000000
451, 0.0000000
452, 0.0000000
453, 0.0000000
454, 0.0000000
455, 0.0000000
456, 0.0000000
457, 0.0000000
458, 0.0000000
459, 0.0000000
460, 0.0000000
461, 0.0000000
462, 0.0000000
463, 0.0000000
464, 0.0000000
465, 0.0000000
466, 0.0000000
467, 0.0000000
468, 0.0000000
469, 0.0000000
470, 0.0000000
471, 0.0000000
472, 0.0000000
473, 0.0000000
474, 0.0000000
475, 0.0000000
476, 0.1000000
477, 0.0000000
478, 0.0000000

}
		data [set=bordes] {x, y
1, 0.7000000
2, 11.5000000
3, 0.9000000
4, 1.3000000
5, 0.6000000
6, 1.4000000
7, 1.4000000
8, 2.0000000
9, 2.6000000
10, 2.1000000
11, 1.5000000
12, 1.7000000
13, 1.8000000
14, 2.2000000
15, 2.0000000
16, 1.3000000
17, 1.6000000
18, 1.9000000
19, 1.9000000
20, 1.1000000
21, 0.9000000
22, 0.9000000
23, 0.9000000
24, 1.3000000
25, 1.1000000
26, 1.4000000
27, 0.9000000
28, 1.7000000
29, 1.1000000
30, 0.8000000
31, 0.7000000
32, 1.1000000
33, 1.0000000
34, 0.9000000
35, 1.1000000
36, 0.7000000
37, 0.5000000
38, 0.7000000
39, 1.4000000
40, 0.6000000
41, 0.5000000
42, 1.2000000
43, 0.8000000
44, 0.5000000
45, 0.9000000
46, 0.5000000
47, 0.4000000
48, 1.0000000
49, 0.8000000
50, 0.5000000
51, 0.8000000
52, 0.5000000
53, 1.1000000
54, 0.2000000
55, 0.3000000
56, 0.6000000
57, 0.3000000
58, 0.5000000
59, 0.8000000
60, 0.6000000
61, 0.3000000
62, 0.4000000
63, 0.5000000
64, 0.5000000
65, 0.4000000
66, 0.5000000
67, 0.4000000
68, 0.4000000
69, 0.6000000
70, 0.0000000
71, 0.3000000
72, 0.6000000
73, 0.7000000
74, 0.7000000
75, 0.2000000
76, 0.4000000
77, 0.3000000
78, 0.4000000
79, 0.6000000
80, 0.8000000
81, 0.3000000
82, 0.2000000
83, 0.6000000
84, 0.2000000
85, 0.1000000
86, 0.4000000
87, 0.3000000
88, 0.1000000
89, 0.1000000
90, 0.2000000
91, 0.3000000
92, 0.2000000
93, 0.4000000
94, 0.4000000
95, 0.2000000
96, 0.3000000
97, 0.1000000
98, 0.2000000
99, 0.3000000
100, 0.3000000
101, 0.4000000
102, 0.2000000
103, 0.2000000
104, 0.1000000
105, 0.2000000
106, 0.1000000
107, 0.4000000
108, 0.2000000
109, 0.0000000
110, 0.2000000
111, 0.2000000
112, 0.3000000
113, 0.0000000
114, 0.4000000
115, 0.3000000
116, 0.2000000
117, 0.5000000
118, 0.2000000
119, 0.0000000
120, 0.0000000
121, 0.3000000
122, 0.1000000
123, 0.2000000
124, 0.1000000
125, 0.1000000
126, 0.1000000
127, 0.0000000
128, 0.1000000
129, 0.0000000
130, 0.2000000
131, 0.0000000
132, 0.1000000
133, 0.1000000
134, 0.1000000
135, 0.5000000
136, 0.1000000
137, 0.0000000
138, 0.4000000
139, 0.0000000
140, 0.1000000
141, 0.0000000
142, 0.1000000
143, 0.1000000
144, 0.1000000
145, 0.1000000
146, 0.0000000
147, 0.0000000
148, 0.0000000
149, 0.0000000
150, 0.0000000
151, 0.2000000
152, 0.0000000
153, 0.1000000
154, 0.3000000
155, 0.1000000
156, 0.0000000
157, 0.0000000
158, 0.0000000
159, 0.0000000
160, 0.0000000
161, 0.1000000
162, 0.1000000
163, 0.3000000
164, 0.0000000
165, 0.2000000
166, 0.0000000
167, 0.0000000
168, 0.0000000
169, 0.0000000
170, 0.0000000
171, 0.0000000
172, 0.3000000
173, 0.0000000
174, 0.0000000
175, 0.1000000
176, 0.0000000
177, 0.1000000
178, 0.0000000
179, 0.1000000
180, 0.0000000
181, 0.0000000
182, 0.2000000
183, 0.1000000
184, 0.1000000
185, 0.0000000
186, 0.0000000
187, 0.0000000
188, 0.1000000
189, 0.2000000
190, 0.0000000
191, 0.0000000
192, 0.0000000
193, 0.0000000
194, 0.0000000
195, 0.0000000
196, 0.1000000
197, 0.1000000
198, 0.0000000
199, 0.0000000
200, 0.0000000
201, 0.1000000
202, 0.0000000
203, 0.0000000
204, 0.0000000
205, 0.1000000
206, 0.0000000
207, 0.1000000
208, 0.0000000
209, 0.2000000
210, 0.1000000
211, 0.0000000
212, 0.1000000
213, 0.0000000
214, 0.0000000
215, 0.0000000
216, 0.1000000
217, 0.0000000
218, 0.0000000
219, 0.0000000
220, 0.0000000
221, 0.0000000
222, 0.1000000
223, 0.1000000
224, 0.0000000
225, 0.0000000
226, 0.0000000
227, 0.0000000
228, 0.0000000
229, 0.0000000
230, 0.0000000
231, 0.0000000
232, 0.1000000
233, 0.0000000
234, 0.1000000
235, 0.0000000
236, 0.0000000
237, 0.0000000
238, 0.0000000
239, 0.0000000
240, 0.0000000
241, 0.0000000
242, 0.0000000
243, 0.0000000
244, 0.1000000
245, 0.0000000
246, 0.0000000
247, 0.0000000
248, 0.0000000
249, 0.0000000
250, 0.0000000
251, 0.0000000
252, 0.0000000
253, 0.0000000
254, 0.0000000
255, 0.0000000
256, 0.0000000
257, 0.0000000
258, 0.0000000
259, 0.0000000
260, 0.1000000
261, 0.0000000
262, 0.0000000
263, 0.0000000
264, 0.0000000
265, 0.0000000
266, 0.1000000
267, 0.0000000
268, 0.0000000
269, 0.0000000
270, 0.0000000
271, 0.0000000
272, 0.0000000
273, 0.0000000
274, 0.0000000
275, 0.0000000
276, 0.0000000
277, 0.0000000
278, 0.0000000
279, 0.0000000
280, 0.0000000
281, 0.0000000
282, 0.0000000
283, 0.0000000
284, 0.0000000
285, 0.0000000
286, 0.0000000
287, 0.0000000
288, 0.0000000
289, 0.0000000
290, 0.0000000
291, 0.0000000
292, 0.0000000
293, 0.0000000
294, 0.0000000
295, 0.0000000
296, 0.0000000
297, 0.0000000
298, 0.0000000
299, 0.0000000
300, 0.0000000
301, 0.0000000
302, 0.0000000
303, 0.0000000
304, 0.0000000
305, 0.0000000
306, 0.0000000
307, 0.0000000
308, 0.1000000
309, 0.0000000
310, 0.0000000
311, 0.0000000
312, 0.0000000
313, 0.0000000
314, 0.0000000
315, 0.0000000
316, 0.0000000
317, 0.0000000
318, 0.0000000
319, 0.0000000
320, 0.0000000
321, 0.0000000
322, 0.0000000
323, 0.0000000
324, 0.0000000
325, 0.0000000
326, 0.0000000
327, 0.0000000
328, 0.0000000
329, 0.0000000
330, 0.0000000
331, 0.0000000
332, 0.0000000
333, 0.0000000
334, 0.0000000
335, 0.0000000
336, 0.0000000
337, 0.0000000
338, 0.0000000
339, 0.0000000
340, 0.0000000
341, 0.0000000
342, 0.0000000
343, 0.0000000
344, 0.0000000
345, 0.0000000
346, 0.0000000
347, 0.0000000
348, 0.0000000
349, 0.0000000
350, 0.0000000
351, 0.0000000
352, 0.0000000
353, 0.0000000
354, 0.0000000
355, 0.0000000
356, 0.0000000
357, 0.0000000
358, 0.0000000
359, 0.0000000
360, 0.0000000
361, 0.0000000
362, 0.0000000
363, 0.0000000
364, 0.0000000
365, 0.0000000
366, 0.0000000
367, 0.0000000
368, 0.0000000
369, 0.0000000
370, 0.0000000
371, 0.0000000
372, 0.0000000
373, 0.0000000
374, 0.0000000
375, 0.0000000
376, 0.0000000
377, 0.0000000
378, 0.0000000
379, 0.0000000
380, 0.0000000
381, 0.0000000
382, 0.0000000
383, 0.0000000
384, 0.0000000
385, 0.0000000
386, 0.0000000
387, 0.0000000
388, 0.0000000
389, 0.0000000
390, 0.0000000
391, 0.0000000
392, 0.0000000
393, 0.0000000
394, 0.0000000
395, 0.1000000
396, 0.0000000
397, 0.0000000
398, 0.0000000
399, 0.0000000
400, 0.0000000
401, 0.0000000
402, 0.0000000
403, 0.0000000
404, 0.0000000
405, 0.0000000
406, 0.0000000
407, 0.0000000
408, 0.0000000
409, 0.0000000
410, 0.0000000
411, 0.0000000
412, 0.0000000
413, 0.0000000
414, 0.0000000
415, 0.0000000
416, 0.0000000
417, 0.0000000
418, 0.0000000
419, 0.0000000
420, 0.0000000
421, 0.0000000
422, 0.0000000
423, 0.0000000
424, 0.0000000
425, 0.0000000
426, 0.0000000
427, 0.0000000
428, 0.0000000
429, 0.0000000
430, 0.0000000
431, 0.0000000
432, 0.0000000
433, 0.0000000
434, 0.0000000
435, 0.0000000
436, 0.0000000
437, 0.0000000
438, 0.0000000
439, 0.0000000
440, 0.0000000
441, 0.0000000
442, 0.0000000
443, 0.0000000
444, 0.0000000
445, 0.0000000
446, 0.0000000
447, 0.0000000
448, 0.0000000
449, 0.0000000
450, 0.0000000
451, 0.0000000
452, 0.0000000
453, 0.0000000
454, 0.0000000
455, 0.0000000
456, 0.0000000
457, 0.0000000
458, 0.0000000
459, 0.0000000
460, 0.0000000
461, 0.0000000
462, 0.0000000
463, 0.0000000
464, 0.0000000
465, 0.0000000
466, 0.0000000
467, 0.0000000
468, 0.0000000
469, 0.0000000
470, 0.0000000
471, 0.0000000
472, 0.0000000
473, 0.0000000
474, 0.0000000
475, 0.0000000
476, 0.1000000
477, 0.0000000
478, 0.0000000

}
		data [set=gillies] {x, y
1, 0.7000000
2, 24.5000000
3, 0.0000000
4, 28.0000000
5, 6.5000000
6, 5.5000000
7, 0.5000000
8, 4.5000000
9, 1.3000000
10, 3.9000000
11, 0.9000000
12, 2.7000000
13, 0.9000000
14, 1.3000000
15, 0.5000000
16, 0.5000000
17, 0.3000000
18, 0.4000000
19, 0.0000000
20, 0.9000000
21, 0.3000000
22, 0.3000000
23, 0.3000000
24, 0.5000000
25, 0.2000000
26, 0.2000000
27, 0.1000000
28, 0.7000000
29, 0.2000000
30, 0.2000000
31, 0.3000000
32, 0.7000000
33, 0.2000000
34, 0.3000000
35, 0.3000000
36, 0.1000000
37, 0.4000000
38, 0.5000000
39, 0.3000000
40, 0.1000000
41, 0.3000000
42, 0.1000000
43, 0.1000000
44, 0.0000000
45, 0.2000000
46, 0.2000000
47, 0.2000000
48, 0.1000000
49, 0.0000000
50, 0.2000000
51, 0.0000000
52, 0.3000000
53, 0.3000000
54, 0.0000000
55, 0.0000000
56, 0.0000000
57, 0.2000000
58, 0.3000000
59, 0.2000000
60, 0.0000000
61, 0.1000000
62, 0.2000000
63, 0.1000000
64, 0.1000000
65, 0.0000000
66, 0.2000000
67, 0.3000000
68, 0.0000000
69, 0.1000000
70, 0.0000000
71, 0.2000000
72, 0.0000000
73, 0.2000000
74, 0.3000000
75, 0.2000000
76, 0.0000000
77, 0.0000000
78, 0.0000000
79, 0.0000000
80, 0.0000000
81, 0.1000000
82, 0.0000000
83, 0.1000000
84, 0.1000000
85, 0.1000000
86, 0.1000000
87, 0.0000000
88, 0.0000000
89, 0.1000000
90, 0.0000000
91, 0.0000000
92, 0.1000000
93, 0.1000000
94, 0.0000000
95, 0.1000000
96, 0.3000000
97, 0.0000000
98, 0.0000000
99, 0.0000000
100, 0.1000000
101, 0.1000000
102, 0.1000000
103, 0.0000000
104, 0.0000000
105, 0.0000000
106, 0.1000000
107, 0.1000000
108, 0.0000000
109, 0.2000000
110, 0.1000000
111, 0.0000000
112, 0.1000000
113, 0.0000000
114, 0.0000000
115, 0.0000000
116, 0.1000000
117, 0.0000000
118, 0.0000000
119, 0.0000000
120, 0.0000000
121, 0.2000000
122, 0.1000000
123, 0.0000000
124, 0.1000000
125, 0.0000000
126, 0.1000000
127, 0.1000000
128, 0.0000000
129, 0.0000000
130, 0.0000000
131, 0.0000000
132, 0.0000000
133, 0.1000000
134, 0.0000000
135, 0.0000000
136, 0.0000000
137, 0.0000000
138, 0.0000000
139, 0.1000000
140, 0.0000000
141, 0.0000000
142, 0.1000000
143, 0.1000000
144, 0.1000000
145, 0.0000000
146, 0.0000000
147, 0.0000000
148, 0.0000000
149, 0.0000000
150, 0.0000000
151, 0.0000000
152, 0.0000000
153, 0.0000000
154, 0.0000000
155, 0.1000000
156, 0.0000000
157, 0.0000000
158, 0.0000000
159, 0.0000000
160, 0.0000000
161, 0.0000000
162, 0.1000000
163, 0.0000000
164, 0.1000000
165, 0.0000000
166, 0.0000000
167, 0.0000000
168, 0.3000000
169, 0.0000000
170, 0.0000000
171, 0.0000000
172, 0.0000000
173, 0.0000000
174, 0.0000000
175, 0.1000000
176, 0.1000000
177, 0.0000000
178, 0.0000000
179, 0.1000000
180, 0.0000000
181, 0.2000000
182, 0.0000000
183, 0.0000000
184, 0.0000000
185, 0.0000000
186, 0.0000000
187, 0.1000000
188, 0.0000000
189, 0.0000000
190, 0.0000000
191, 0.0000000
192, 0.0000000
193, 0.0000000
194, 0.0000000
195, 0.0000000
196, 0.0000000
197, 0.0000000
198, 0.0000000
199, 0.0000000
200, 0.0000000
201, 0.0000000
202, 0.0000000
203, 0.0000000
204, 0.0000000
205, 0.0000000
206, 0.0000000
207, 0.0000000
208, 0.0000000
209, 0.0000000
210, 0.0000000
211, 0.0000000
212, 0.1000000
213, 0.0000000
214, 0.0000000
215, 0.0000000
216, 0.0000000
217, 0.0000000
218, 0.1000000
219, 0.1000000
220, 0.0000000
221, 0.0000000
222, 0.0000000
223, 0.0000000
224, 0.0000000
225, 0.0000000
226, 0.1000000
227, 0.0000000
228, 0.0000000
229, 0.0000000
230, 0.0000000
231, 0.0000000
232, 0.0000000
233, 0.0000000
234, 0.0000000
235, 0.0000000
236, 0.0000000
237, 0.0000000
238, 0.0000000
239, 0.0000000
240, 0.0000000
241, 0.0000000
242, 0.0000000
243, 0.0000000
244, 0.0000000
245, 0.0000000
246, 0.0000000
247, 0.0000000
248, 0.0000000
249, 0.0000000
250, 0.0000000
251, 0.0000000
252, 0.0000000
253, 0.0000000
254, 0.0000000
255, 0.0000000
256, 0.0000000
257, 0.0000000
258, 0.0000000
259, 0.0000000
260, 0.1000000
261, 0.0000000
262, 0.0000000
263, 0.0000000
264, 0.0000000
265, 0.0000000
266, 0.0000000
267, 0.0000000
268, 0.0000000
269, 0.0000000
270, 0.0000000
271, 0.0000000
272, 0.0000000
273, 0.0000000
274, 0.0000000
275, 0.1000000
276, 0.0000000
277, 0.0000000
278, 0.0000000
279, 0.0000000
280, 0.0000000
281, 0.0000000
282, 0.0000000
283, 0.0000000
284, 0.0000000
285, 0.0000000
286, 0.0000000
287, 0.0000000
288, 0.0000000
289, 0.0000000
290, 0.0000000
291, 0.0000000
292, 0.0000000
293, 0.0000000
294, 0.0000000
295, 0.0000000
296, 0.0000000
297, 0.0000000
298, 0.0000000
299, 0.0000000
300, 0.0000000
301, 0.0000000
302, 0.0000000
303, 0.0000000
304, 0.0000000
305, 0.0000000
306, 0.0000000
307, 0.0000000
308, 0.0000000
309, 0.0000000
310, 0.0000000
311, 0.0000000
312, 0.0000000
313, 0.0000000
314, 0.0000000
315, 0.0000000
316, 0.0000000
317, 0.0000000
318, 0.0000000
319, 0.0000000
320, 0.0000000
321, 0.0000000
322, 0.0000000
323, 0.0000000
324, 0.0000000
325, 0.0000000
326, 0.0000000
327, 0.0000000
328, 0.0000000
329, 0.0000000
330, 0.0000000
331, 0.0000000
332, 0.0000000
333, 0.0000000
334, 0.0000000
335, 0.0000000
336, 0.0000000
337, 0.0000000
338, 0.0000000
339, 0.0000000
340, 0.0000000
341, 0.0000000
342, 0.0000000
343, 0.0000000
344, 0.0000000
345, 0.0000000
346, 0.0000000
347, 0.0000000
348, 0.0000000
349, 0.0000000
350, 0.0000000
351, 0.0000000
352, 0.0000000
353, 0.0000000
354, 0.0000000
355, 0.0000000
356, 0.0000000
357, 0.0000000
358, 0.0000000
359, 0.0000000
360, 0.0000000
361, 0.0000000
362, 0.0000000
363, 0.0000000
364, 0.0000000
365, 0.0000000
366, 0.0000000
367, 0.0000000
368, 0.0000000
369, 0.0000000
370, 0.0000000
371, 0.0000000
372, 0.0000000
373, 0.0000000
374, 0.0000000
375, 0.0000000
376, 0.0000000
377, 0.0000000
378, 0.0000000
379, 0.0000000
380, 0.0000000
381, 0.0000000
382, 0.0000000
383, 0.0000000
384, 0.0000000
385, 0.0000000
386, 0.0000000
387, 0.0000000
388, 0.0000000
389, 0.0000000
390, 0.0000000
391, 0.0000000
392, 0.0000000
393, 0.0000000
394, 0.1000000
395, 0.0000000
396, 0.1000000
397, 0.0000000
398, 0.0000000
399, 0.0000000
400, 0.0000000
401, 0.0000000
402, 0.0000000
403, 0.0000000
404, 0.0000000
405, 0.0000000
406, 0.0000000
407, 0.0000000
408, 0.0000000
409, 0.0000000
410, 0.0000000
411, 0.0000000
412, 0.0000000
413, 0.0000000
414, 0.0000000
415, 0.0000000
416, 0.0000000
417, 0.0000000
418, 0.0000000
419, 0.0000000
420, 0.0000000
421, 0.0000000
422, 0.0000000
423, 0.0000000
424, 0.0000000
425, 0.0000000
426, 0.0000000
427, 0.0000000
428, 0.0000000
429, 0.0000000
430, 0.0000000
431, 0.0000000
432, 0.0000000
433, 0.0000000
434, 0.0000000
435, 0.0000000
436, 0.0000000
437, 0.0000000
438, 0.0000000
439, 0.0000000
440, 0.0000000
441, 0.0000000
442, 0.0000000
443, 0.0000000
444, 0.0000000
445, 0.0000000
446, 0.0000000
447, 0.0000000
448, 0.0000000
449, 0.0000000
450, 0.0000000
451, 0.0000000
452, 0.0000000
453, 0.0000000
454, 0.0000000
455, 0.0000000
456, 0.0000000
457, 0.0000000
458, 0.0000000
459, 0.0000000
460, 0.0000000
461, 0.0000000
462, 0.0000000
463, 0.0000000
464, 0.0000000
465, 0.0000000
466, 0.0000000
467, 0.0000000
468, 0.0000000
469, 0.0000000
470, 0.0000000
471, 0.0000000
472, 0.0000000
473, 0.0000000
474, 0.0000000
475, 0.0000000
476, 0.0000000
477, 0.0000000
478, 0.0000000

}
		data [set=po] {x, y
1, 0.7000000
2, 0.4000000
3, 0.1000000
4, 0.0000000
5, 0.0000000
6, 0.1000000
7, 0.6000000
8, 0.4000000
9, 0.3000000
10, 0.6000000
11, 0.4000000
12, 0.5000000
13, 0.7000000
14, 0.7000000
15, 0.5000000
16, 0.6000000
17, 0.8000000
18, 0.9000000
19, 0.9000000
20, 0.7000000
21, 0.9000000
22, 0.5000000
23, 1.2000000
24, 0.9000000
25, 1.0000000
26, 1.5000000
27, 0.5000000
28, 1.0000000
29, 0.6000000
30, 1.5000000
31, 0.8000000
32, 0.5000000
33, 1.0000000
34, 1.5000000
35, 1.2000000
36, 0.9000000
37, 0.7000000
38, 1.5000000
39, 1.1000000
40, 1.7000000
41, 0.8000000
42, 0.8000000
43, 0.8000000
44, 0.7000000
45, 1.1000000
46, 1.4000000
47, 1.2000000
48, 0.7000000
49, 0.9000000
50, 1.2000000
51, 1.2000000
52, 1.5000000
53, 1.3000000
54, 1.4000000
55, 0.7000000
56, 0.8000000
57, 1.3000000
58, 1.4000000
59, 0.5000000
60, 0.8000000
61, 1.0000000
62, 1.0000000
63, 1.3000000
64, 1.0000000
65, 0.8000000
66, 1.1000000
67, 1.2000000
68, 0.8000000
69, 0.9000000
70, 0.4000000
71, 1.0000000
72, 1.2000000
73, 0.7000000
74, 0.9000000
75, 1.3000000
76, 0.9000000
77, 0.8000000
78, 0.4000000
79, 0.8000000
80, 0.4000000
81, 0.5000000
82, 0.9000000
83, 0.4000000
84, 0.7000000
85, 0.7000000
86, 0.4000000
87, 0.6000000
88, 0.8000000
89, 0.7000000
90, 0.9000000
91, 0.9000000
92, 0.0000000
93, 0.3000000
94, 0.3000000
95, 1.0000000
96, 1.1000000
97, 0.5000000
98, 0.5000000
99, 0.5000000
100, 0.4000000
101, 0.8000000
102, 0.5000000
103, 0.4000000
104, 0.4000000
105, 0.9000000
106, 0.2000000
107, 0.3000000
108, 0.5000000
109, 0.7000000
110, 0.3000000
111, 0.1000000
112, 0.2000000
113, 0.3000000
114, 0.3000000
115, 0.2000000
116, 0.5000000
117, 0.1000000
118, 0.3000000
119, 0.2000000
120, 0.8000000
121, 0.3000000
122, 0.1000000
123, 0.1000000
124, 0.3000000
125, 0.3000000
126, 0.0000000
127, 0.2000000
128, 0.3000000
129, 0.2000000
130, 0.0000000
131, 0.2000000
132, 0.0000000
133, 0.0000000
134, 0.1000000
135, 0.1000000
136, 0.4000000
137, 0.4000000
138, 0.2000000
139, 0.3000000
140, 0.1000000
141, 0.2000000
142, 0.2000000
143, 0.4000000
144, 0.2000000
145, 0.2000000
146, 0.1000000
147, 0.2000000
148, 0.0000000
149, 0.0000000
150, 0.0000000
151, 0.2000000
152, 0.0000000
153, 0.1000000
154, 0.2000000
155, 0.1000000
156, 0.2000000
157, 0.1000000
158, 0.0000000
159, 0.1000000
160, 0.1000000
161, 0.1000000
162, 0.1000000
163, 0.3000000
164, 0.0000000
165, 0.0000000
166, 0.1000000
167, 0.1000000
168, 0.1000000
169, 0.3000000
170, 0.0000000
171, 0.0000000
172, 0.1000000
173, 0.0000000
174, 0.1000000
175, 0.1000000
176, 0.1000000
177, 0.0000000
178, 0.0000000
179, 0.1000000
180, 0.0000000
181, 0.1000000
182, 0.0000000
183, 0.1000000
184, 0.1000000
185, 0.0000000
186, 0.0000000
187, 0.0000000
188, 0.2000000
189, 0.1000000
190, 0.1000000
191, 0.0000000
192, 0.0000000
193, 0.1000000
194, 0.3000000
195, 0.0000000
196, 0.1000000
197, 0.0000000
198, 0.1000000
199, 0.0000000
200, 0.0000000
201, 0.2000000
202, 0.0000000
203, 0.0000000
204, 0.0000000
205, 0.0000000
206, 0.0000000
207, 0.1000000
208, 0.1000000
209, 0.2000000
210, 0.1000000
211, 0.0000000
212, 0.1000000
213, 0.0000000
214, 0.1000000
215, 0.0000000
216, 0.0000000
217, 0.0000000
218, 0.0000000
219, 0.0000000
220, 0.0000000
221, 0.1000000
222, 0.0000000
223, 0.1000000
224, 0.0000000
225, 0.0000000
226, 0.0000000
227, 0.0000000
228, 0.1000000
229, 0.0000000
230, 0.0000000
231, 0.0000000
232, 0.0000000
233, 0.0000000
234, 0.0000000
235, 0.0000000
236, 0.1000000
237, 0.0000000
238, 0.0000000
239, 0.0000000
240, 0.0000000
241, 0.1000000
242, 0.0000000
243, 0.0000000
244, 0.0000000
245, 0.0000000
246, 0.0000000
247, 0.1000000
248, 0.0000000
249, 0.0000000
250, 0.0000000
251, 0.0000000
252, 0.0000000
253, 0.0000000
254, 0.0000000
255, 0.0000000
256, 0.0000000
257, 0.0000000
258, 0.0000000
259, 0.0000000
260, 0.1000000
261, 0.0000000
262, 0.0000000
263, 0.0000000
264, 0.0000000
265, 0.0000000
266, 0.1000000
267, 0.0000000
268, 0.0000000
269, 0.0000000
270, 0.0000000
271, 0.0000000
272, 0.0000000
273, 0.0000000
274, 0.0000000
275, 0.0000000
276, 0.0000000
277, 0.0000000
278, 0.0000000
279, 0.0000000
280, 0.0000000
281, 0.0000000
282, 0.0000000
283, 0.0000000
284, 0.0000000
285, 0.0000000
286, 0.0000000
287, 0.0000000
288, 0.0000000
289, 0.0000000
290, 0.0000000
291, 0.0000000
292, 0.0000000
293, 0.0000000
294, 0.0000000
295, 0.0000000
296, 0.0000000
297, 0.0000000
298, 0.0000000
299, 0.0000000
300, 0.0000000
301, 0.0000000
302, 0.0000000
303, 0.0000000
304, 0.0000000
305, 0.0000000
306, 0.0000000
307, 0.0000000
308, 0.0000000
309, 0.0000000
310, 0.0000000
311, 0.1000000
312, 0.0000000
313, 0.0000000
314, 0.0000000
315, 0.0000000
316, 0.0000000
317, 0.0000000
318, 0.0000000
319, 0.0000000
320, 0.0000000
321, 0.0000000
322, 0.0000000
323, 0.0000000
324, 0.0000000
325, 0.0000000
326, 0.0000000
327, 0.0000000
328, 0.0000000
329, 0.0000000
330, 0.0000000
331, 0.0000000
332, 0.0000000
333, 0.0000000
334, 0.0000000
335, 0.0000000
336, 0.0000000
337, 0.0000000
338, 0.0000000
339, 0.0000000
340, 0.0000000
341, 0.0000000
342, 0.0000000
343, 0.0000000
344, 0.0000000
345, 0.0000000
346, 0.0000000
347, 0.0000000
348, 0.0000000
349, 0.0000000
350, 0.0000000
351, 0.0000000
352, 0.0000000
353, 0.0000000
354, 0.0000000
355, 0.0000000
356, 0.0000000
357, 0.0000000
358, 0.0000000
359, 0.0000000
360, 0.0000000
361, 0.0000000
362, 0.0000000
363, 0.0000000
364, 0.0000000
365, 0.0000000
366, 0.0000000
367, 0.0000000
368, 0.0000000
369, 0.0000000
370, 0.0000000
371, 0.0000000
372, 0.0000000
373, 0.0000000
374, 0.0000000
375, 0.0000000
376, 0.0000000
377, 0.0000000
378, 0.0000000
379, 0.0000000
380, 0.0000000
381, 0.0000000
382, 0.0000000
383, 0.0000000
384, 0.0000000
385, 0.0000000
386, 0.0000000
387, 0.0000000
388, 0.0000000
389, 0.0000000
390, 0.0000000
391, 0.0000000
392, 0.0000000
393, 0.0000000
394, 0.0000000
395, 0.1000000
396, 0.0000000
397, 0.0000000
398, 0.0000000
399, 0.0000000
400, 0.0000000
401, 0.0000000
402, 0.0000000
403, 0.0000000
404, 0.0000000
405, 0.0000000
406, 0.0000000
407, 0.0000000
408, 0.0000000
409, 0.0000000
410, 0.0000000
411, 0.0000000
412, 0.0000000
413, 0.0000000
414, 0.0000000
415, 0.0000000
416, 0.0000000
417, 0.0000000
418, 0.0000000
419, 0.0000000
420, 0.0000000
421, 0.0000000
422, 0.0000000
423, 0.0000000
424, 0.0000000
425, 0.0000000
426, 0.0000000
427, 0.0000000
428, 0.0000000
429, 0.0000000
430, 0.0000000
431, 0.0000000
432, 0.0000000
433, 0.0000000
434, 0.0000000
435, 0.0000000
436, 0.0000000
437, 0.0000000
438, 0.0000000
439, 0.0000000
440, 0.0000000
441, 0.0000000
442, 0.0000000
443, 0.0000000
444, 0.0000000
445, 0.0000000
446, 0.0000000
447, 0.0000000
448, 0.0000000
449, 0.0000000
450, 0.0000000
451, 0.0000000
452, 0.0000000
453, 0.0000000
454, 0.0000000
455, 0.0000000
456, 0.0000000
457, 0.0000000
458, 0.0000000
459, 0.0000000
460, 0.0000000
461, 0.0000000
462, 0.0000000
463, 0.0000000
464, 0.0000000
465, 0.0000000
466, 0.0000000
467, 0.0000000
468, 0.0000000
469, 0.0000000
470, 0.0000000
471, 0.0000000
472, 0.0000000
473, 0.0000000
474, 0.0000000
475, 0.0000000
476, 0.0000000
477, 0.0000000
478, 0.1000000

};
		\begin{scope}[xshift=35,yshift=20,xscale=0.5,yscale=0.7]
\clip (-1,-1) -- (-1,3.2) -- (-0.1,3.2) -- (-0.1,3.1) -- (6,3.1) -- (6,-1) -- (-1,-1);
\datavisualization [scientific axes=clean,
x axis={max value=100},
y axis={max value=5},
visualize as line/.list={mckelvey, bordes, gillies, po},
		style sheet=strong colors
]
		data [set=mckelvey] {x, y
			1, 0.7000000
			2, 11.5000000
			3, 0.9000000
			4, 1.3000000
			5, 0.6000000
			6, 1.4000000
			7, 1.4000000
			8, 2.0000000
			9, 2.6000000
			10, 2.1000000
			11, 1.5000000
			12, 1.7000000
			13, 1.8000000
			14, 2.1000000
			15, 2.0000000
			16, 1.4000000
			17, 1.6000000
			18, 1.9000000
			19, 1.8000000
			20, 1.2000000
			21, 0.9000000
			22, 0.8000000
			23, 1.0000000
			24, 1.2000000
			25, 1.2000000
			26, 1.4000000
			27, 0.9000000
			28, 1.7000000
			29, 1.1000000
			30, 0.8000000
			31, 0.7000000
			32, 1.1000000
			33, 0.9000000
			34, 1.0000000
			35, 1.1000000
			36, 0.7000000
			37, 0.5000000
			38, 0.7000000
			39, 1.4000000
			40, 0.5000000
			41, 0.6000000
			42, 1.2000000
			43, 0.8000000
			44, 0.5000000
			45, 0.9000000
			46, 0.5000000
			47, 0.4000000
			48, 1.0000000
			49, 0.8000000
			50, 0.5000000
			51, 0.8000000
			52, 0.5000000
			53, 1.1000000
			54, 0.2000000
			55, 0.3000000
			56, 0.6000000
			57, 0.3000000
			58, 0.5000000
			59, 0.8000000
			60, 0.6000000
			61, 0.3000000
			62, 0.4000000
			63, 0.5000000
			64, 0.5000000
			65, 0.4000000
			66, 0.4000000
			67, 0.5000000
			68, 0.4000000
			69, 0.6000000
			70, 0.0000000
			71, 0.3000000
			72, 0.6000000
			73, 0.7000000
			74, 0.7000000
			75, 0.2000000
			76, 0.4000000
			77, 0.3000000
			78, 0.4000000
			79, 0.6000000
			80, 0.8000000
			81, 0.3000000
			82, 0.2000000
			83, 0.6000000
			84, 0.2000000
			85, 0.1000000
			86, 0.4000000
			87, 0.3000000
			88, 0.1000000
			89, 0.1000000
			90, 0.2000000
			91, 0.3000000
			92, 0.2000000
			93, 0.4000000
			94, 0.4000000
			95, 0.2000000
			96, 0.3000000
			97, 0.1000000
			98, 0.2000000
			99, 0.2000000
			100, 0.4000000
		}
		data [set=bordes] {x, y
			1, 0.7000000
			2, 11.5000000
			3, 0.9000000
			4, 1.3000000
			5, 0.6000000
			6, 1.4000000
			7, 1.4000000
			8, 2.0000000
			9, 2.6000000
			10, 2.1000000
			11, 1.5000000
			12, 1.7000000
			13, 1.8000000
			14, 2.2000000
			15, 2.0000000
			16, 1.3000000
			17, 1.6000000
			18, 1.9000000
			19, 1.9000000
			20, 1.1000000
			21, 0.9000000
			22, 0.9000000
			23, 0.9000000
			24, 1.3000000
			25, 1.1000000
			26, 1.4000000
			27, 0.9000000
			28, 1.7000000
			29, 1.1000000
			30, 0.8000000
			31, 0.7000000
			32, 1.1000000
			33, 1.0000000
			34, 0.9000000
			35, 1.1000000
			36, 0.7000000
			37, 0.5000000
			38, 0.7000000
			39, 1.4000000
			40, 0.6000000
			41, 0.5000000
			42, 1.2000000
			43, 0.8000000
			44, 0.5000000
			45, 0.9000000
			46, 0.5000000
			47, 0.4000000
			48, 1.0000000
			49, 0.8000000
			50, 0.5000000
			51, 0.8000000
			52, 0.5000000
			53, 1.1000000
			54, 0.2000000
			55, 0.3000000
			56, 0.6000000
			57, 0.3000000
			58, 0.5000000
			59, 0.8000000
			60, 0.6000000
			61, 0.3000000
			62, 0.4000000
			63, 0.5000000
			64, 0.5000000
			65, 0.4000000
			66, 0.5000000
			67, 0.4000000
			68, 0.4000000
			69, 0.6000000
			70, 0.0000000
			71, 0.3000000
			72, 0.6000000
			73, 0.7000000
			74, 0.7000000
			75, 0.2000000
			76, 0.4000000
			77, 0.3000000
			78, 0.4000000
			79, 0.6000000
			80, 0.8000000
			81, 0.3000000
			82, 0.2000000
			83, 0.6000000
			84, 0.2000000
			85, 0.1000000
			86, 0.4000000
			87, 0.3000000
			88, 0.1000000
			89, 0.1000000
			90, 0.2000000
			91, 0.3000000
			92, 0.2000000
			93, 0.4000000
			94, 0.4000000
			95, 0.2000000
			96, 0.3000000
			97, 0.1000000
			98, 0.2000000
			99, 0.3000000
			100, 0.3000000
		}
		data [set=gillies] {x, y
			1, 0.7000000
			2, 24.5000000
			3, 0.0000000
			4, 28.0000000
			5, 6.5000000
			6, 5.5000000
			7, 0.5000000
			8, 4.5000000
			9, 1.3000000
			10, 3.9000000
			11, 0.9000000
			12, 2.7000000
			13, 0.9000000
			14, 1.3000000
			15, 0.5000000
			16, 0.5000000
			17, 0.3000000
			18, 0.4000000
			19, 0.0000000
			20, 0.9000000
			21, 0.3000000
			22, 0.3000000
			23, 0.3000000
			24, 0.5000000
			25, 0.2000000
			26, 0.2000000
			27, 0.1000000
			28, 0.7000000
			29, 0.2000000
			30, 0.2000000
			31, 0.3000000
			32, 0.7000000
			33, 0.2000000
			34, 0.3000000
			35, 0.3000000
			36, 0.1000000
			37, 0.4000000
			38, 0.5000000
			39, 0.3000000
			40, 0.1000000
			41, 0.3000000
			42, 0.1000000
			43, 0.1000000
			44, 0.0000000
			45, 0.2000000
			46, 0.2000000
			47, 0.2000000
			48, 0.1000000
			49, 0.0000000
			50, 0.2000000
			51, 0.0000000
			52, 0.3000000
			53, 0.3000000
			54, 0.0000000
			55, 0.0000000
			56, 0.0000000
			57, 0.2000000
			58, 0.3000000
			59, 0.2000000
			60, 0.0000000
			61, 0.1000000
			62, 0.2000000
			63, 0.1000000
			64, 0.1000000
			65, 0.0000000
			66, 0.2000000
			67, 0.3000000
			68, 0.0000000
			69, 0.1000000
			70, 0.0000000
			71, 0.2000000
			72, 0.0000000
			73, 0.2000000
			74, 0.3000000
			75, 0.2000000
			76, 0.0000000
			77, 0.0000000
			78, 0.0000000
			79, 0.0000000
			80, 0.0000000
			81, 0.1000000
			82, 0.0000000
			83, 0.1000000
			84, 0.1000000
			85, 0.1000000
			86, 0.1000000
			87, 0.0000000
			88, 0.0000000
			89, 0.1000000
			90, 0.0000000
			91, 0.0000000
			92, 0.1000000
			93, 0.1000000
			94, 0.0000000
			95, 0.1000000
			96, 0.3000000
			97, 0.0000000
			98, 0.0000000
			99, 0.0000000
			100, 0.1000000
		}
		data [set=po] {x, y
			1, 0.7000000
			2, 0.4000000
			3, 0.1000000
			4, 0.0000000
			5, 0.0000000
			6, 0.1000000
			7, 0.6000000
			8, 0.4000000
			9, 0.3000000
			10, 0.6000000
			11, 0.4000000
			12, 0.5000000
			13, 0.7000000
			14, 0.7000000
			15, 0.5000000
			16, 0.6000000
			17, 0.8000000
			18, 0.9000000
			19, 0.9000000
			20, 0.7000000
			21, 0.9000000
			22, 0.5000000
			23, 1.2000000
			24, 0.9000000
			25, 1.0000000
			26, 1.5000000
			27, 0.5000000
			28, 1.0000000
			29, 0.6000000
			30, 1.5000000
			31, 0.8000000
			32, 0.5000000
			33, 1.0000000
			34, 1.5000000
			35, 1.2000000
			36, 0.9000000
			37, 0.7000000
			38, 1.5000000
			39, 1.1000000
			40, 1.7000000
			41, 0.8000000
			42, 0.8000000
			43, 0.8000000
			44, 0.7000000
			45, 1.1000000
			46, 1.4000000
			47, 1.2000000
			48, 0.7000000
			49, 0.9000000
			50, 1.2000000
			51, 1.2000000
			52, 1.5000000
			53, 1.3000000
			54, 1.4000000
			55, 0.7000000
			56, 0.8000000
			57, 1.3000000
			58, 1.4000000
			59, 0.5000000
			60, 0.8000000
			61, 1.0000000
			62, 1.0000000
			63, 1.3000000
			64, 1.0000000
			65, 0.8000000
			66, 1.1000000
			67, 1.2000000
			68, 0.8000000
			69, 0.9000000
			70, 0.4000000
			71, 1.0000000
			72, 1.2000000
			73, 0.7000000
			74, 0.9000000
			75, 1.3000000
			76, 0.9000000
			77, 0.8000000
			78, 0.4000000
			79, 0.8000000
			80, 0.4000000
			81, 0.5000000
			82, 0.9000000
			83, 0.4000000
			84, 0.7000000
			85, 0.7000000
			86, 0.4000000
			87, 0.6000000
			88, 0.8000000
			89, 0.7000000
			90, 0.9000000
			91, 0.9000000
			92, 0.0000000
			93, 0.3000000
			94, 0.3000000
			95, 1.0000000
			96, 1.1000000
			97, 0.5000000
			98, 0.5000000
			99, 0.5000000
			100, 0.4000000
		};
	\draw[->] (0,0) -- (-0.7,-0.6);
	\end{scope}
	\end{tikzpicture}
	\caption{Size distributions of UCs for $n = 7$ in $1,000$ profiles sampled via the impartial culture model. The high peak is at size $2$ for McKelvey and Bordes. Gillies-\uc has an even higher peak at size~$4$.
	}\label{fig:uc7}
    \Description{A plot showing data for the four assignment rules McKelvey uncovered set, Bordes uncovered set, Gillies uncovered set, and the Pareto rule. It depicts the frequency with which each cardinality occurs across one thousand profiles sampled uniformly at random for seven agents and houses. It can be seen that a size of two (and small sizes in general) occurs way more often for the uncovered set than for the Pareto rule. In general, the uncovered sets often return fewer outcomes than the Pareto rule. This is particularly true for the Gillies uncovered set. The Gillies uncovered set has size four even more often than size two.}
\end{figure}
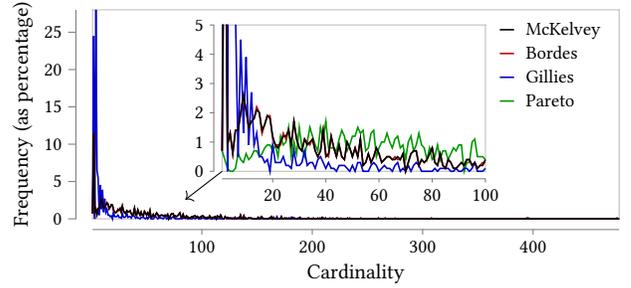
 
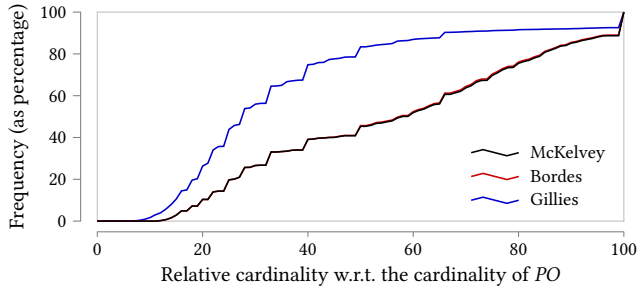
\begin{figure}
	\begin{tikzpicture}[xscale=1.4,yscale=0.9]
		\datavisualization [scientific axes=clean,
		x axis={label={Relative cardinality w.r.t.\ the cardinality of \po}},
		y axis={label={Frequency (as percentage)}},
legend=south east inside,
		visualize as line/.list={mckelvey, bordes, gillies},
		style sheet=strong colors,
		mckelvey= {label in legend={text=McKelvey}},
		bordes= {label in legend={text=Bordes}},
		gillies= {label in legend={text=Gillies}}]
data [set=mckelvey] {x, y
0, 0.0000000
1, 0.0000000
2, 0.0000000
3, 0.0000000
4, 0.0000000
5, 0.0000000
6, 0.0000000
7, 0.0000000
8, 0.0000000
9, 0.0000000
10, 0.0000000
11, 0.0000000
12, 0.2211347
13, 0.6828784
14, 1.5363992
15, 2.8669965
16, 4.9010919
17, 4.9010919
18, 7.1921865
19, 7.1921865
20, 10.3661456
21, 10.3661456
22, 13.9587581
23, 14.3473630
24, 14.3473630
25, 19.7474289
26, 20.0699445
27, 21.0123774
28, 25.6807470
29, 25.7217223
30, 26.6152492
31, 26.7421406
32, 26.7421406
33, 33.0684255
34, 33.0684255
35, 33.2653715
36, 33.4531972
37, 33.9100062
38, 34.0104619
39, 34.0104619
40, 39.0955684
41, 39.2570024
42, 39.6649715
43, 39.7535311
44, 39.9782787
45, 40.0363932
46, 40.4772747
47, 40.8308522
48, 40.8318436
49, 40.8318436
50, 45.3403873
51, 45.3417090
52, 45.8420048
53, 46.7493994
54, 46.9499142
55, 47.4733963
56, 47.9543059
57, 49.3061508
58, 50.0856737
59, 50.1138718
60, 51.8952199
61, 52.8951064
62, 53.5443894
63, 54.7712265
64, 55.7562319
65, 56.0428391
66, 60.6265152
67, 60.6293791
68, 61.3984819
69, 62.1069479
70, 63.7586949
71, 64.7515870
72, 66.5094403
73, 67.3111141
74, 67.3300597
75, 69.6773412
76, 70.9403511
77, 72.4222708
78, 73.3991693
79, 73.5242432
80, 75.5030990
81, 76.3331472
82, 76.9361677
83, 78.0406956
84, 78.8345709
85, 80.7825630
86, 81.6764864
87, 82.6161436
88, 83.7241742
89, 84.0179961
90, 85.1443114
91, 85.6034886
92, 86.3583823
93, 86.8609801
94, 87.3435309
95, 88.0406735
96, 88.5998328
97, 88.6869054
98, 88.6872909
99, 88.6872909
100, 100.0000000

}
		data [set=bordes] {x, y
0, 0.0000000
1, 0.0000000
2, 0.0000000
3, 0.0000000
4, 0.0000000
5, 0.0000000
6, 0.0000000
7, 0.0000000
8, 0.0000000
9, 0.0000000
10, 0.0000000
11, 0.0000000
12, 0.2211347
13, 0.6828784
14, 1.5363992
15, 2.8669965
16, 4.9010919
17, 4.9010919
18, 7.1921865
19, 7.1921865
20, 10.3661456
21, 10.3661456
22, 13.9587581
23, 14.3473630
24, 14.3473630
25, 19.7474289
26, 20.0699445
27, 21.0123774
28, 25.6807470
29, 25.7217223
30, 26.6152492
31, 26.7421406
32, 26.7421406
33, 33.0684255
34, 33.0684255
35, 33.2653715
36, 33.4531972
37, 33.9168795
38, 34.0173352
39, 34.0173352
40, 39.1024417
41, 39.2638757
42, 39.6961656
43, 39.7900124
44, 40.0926572
45, 40.1525561
46, 40.6163485
47, 41.0114522
48, 41.0124435
49, 41.0124435
50, 45.6879727
51, 45.6892945
52, 46.2392674
53, 47.1599239
54, 47.3725111
55, 47.8694252
56, 48.3413026
57, 49.7105731
58, 50.5027411
59, 50.5310493
60, 52.3802270
61, 53.3735046
62, 53.9841474
63, 55.2273746
64, 56.2337489
65, 56.5421655
66, 61.2063164
67, 61.2102817
68, 61.9992774
69, 62.7696690
70, 64.4367707
71, 65.3838299
72, 67.1665879
73, 67.9591414
74, 67.9856322
75, 70.3347311
76, 71.5377981
77, 72.9964103
78, 73.9784086
79, 74.1067320
80, 76.0878789
81, 76.9116155
82, 77.4944898
83, 78.5784199
84, 79.3542968
85, 81.2114934
86, 82.0636814
87, 82.9970271
88, 84.0757031
89, 84.3868954
90, 85.5125718
91, 85.9673100
92, 86.7253539
93, 87.2333931
94, 87.7062508
95, 88.3737194
96, 88.8832897
97, 88.9795597
98, 88.9802757
99, 88.9802757
100, 100.0000000

}
		data [set=gillies] {x, y
0, 0.0000000
1, 0.0000000
2, 0.0000000
3, 0.0000000
4, 0.0000000
5, 0.0003855
6, 0.0153107
7, 0.1134422
8, 0.3739441
9, 0.8700982
10, 1.6696572
11, 2.9848336
12, 3.9534599
13, 5.6913653
14, 7.8539163
15, 10.5782040
16, 14.4754880
17, 14.9148495
18, 19.5957099
19, 20.2753279
20, 26.3347003
21, 27.6884948
22, 34.0546757
23, 35.6643458
24, 35.7711901
25, 43.8016309
26, 45.7489511
27, 46.2696024
28, 53.8304898
29, 54.1293235
30, 55.9408964
31, 56.3276287
32, 56.3509803
33, 64.5620650
34, 64.6305775
35, 64.9556266
36, 66.6999316
37, 67.0808261
38, 67.3839335
39, 67.4743656
40, 74.8254307
41, 75.0169574
42, 75.8368168
43, 75.9417996
44, 77.2890844
45, 77.6288383
46, 77.8860357
47, 78.4438842
48, 78.5345366
49, 78.5345366
50, 83.3606282
51, 83.3997310
52, 83.7667137
53, 84.2097761
54, 84.4528745
55, 84.6784812
56, 84.9404260
57, 86.1167819
58, 86.2847919
59, 86.4118485
60, 86.9452660
61, 87.1932990
62, 87.3354570
63, 87.4665340
64, 87.5963664
65, 87.7139502
66, 90.2865410
67, 90.3107738
68, 90.3913256
69, 90.4904815
70, 90.6400415
71, 90.7033440
72, 90.8332425
73, 90.9024049
74, 90.9276289
75, 91.1231210
76, 91.1959734
77, 91.2479086
78, 91.3095037
79, 91.3303769
80, 91.5790709
81, 91.6302350
82, 91.6601404
83, 91.7005980
84, 91.7398771
85, 91.8291306
86, 91.8617346
87, 91.9034590
88, 91.9550196
89, 91.9713988
90, 92.0512566
91, 92.0998102
92, 92.1823447
93, 92.2714110
94, 92.3567102
95, 92.4055281
96, 92.4854191
97, 92.5450095
98, 92.5454501
99, 92.5454501
100, 100.0000000

};
	\end{tikzpicture}
	\caption{Cumulative distributions of the
	ratio of UC and PO sizes for $n = 5$. The
	plot shows
	the relative difference, i.e., the ratio.
    In total, there are $9,078,630$ profiles up to symmetry.
    We note that the plots for McKelvey and Bordes almost perfectly align.
    }\label{fig:ucPOdiffquot}
    \Description{A plot showing data for the three assignment rules McKelvey uncovered set, Bordes uncovered set, Gillies uncovered set, comparing their sizes with the Pareto rule. It depicts the frequency with which each relative cardinality with respect to the Pareto rule occurs across all profiles for five agents and houses. It can be seen that often the uncovered sets have between ten and fifty percent of the size of the set of all Pareto-optimal assignments. In general, the uncovered sets often return fewer outcomes than the Pareto rule. This is particularly true for the Gillies uncovered set. Between ten and one hundred percent, many different ratios occur regularly.}
\end{figure}
 
\begin{remark}
In social choice, \uc can be computed in polynomial time via matrix multiplication \citep{BFH09b,BBH15a}. In assignment, the uncovered set can be exponentially large. Unless a structural result such as \Cref{thm:Structure_TC} also holds for \uc, it seems unlikely that the uncovered set can be returned in polynomial time. Instead, the two interesting questions are \emph{(i)} whether, given a profile, one can efficiently find an uncovered assignment, and \emph{(ii)} whether, given a profile and an assignment, the assignment is uncovered.
The computational complexity of both problems remains open.
\end{remark}

\begin{remark}
Our experiments suggest that, at least for small $n$, all rank-maximal assignments are contained in the uncovered set. We have verified this through exhaustive search for all profiles with $n\le 5$.
If this were true in general, one could obtain an element of the uncovered set by computing a rank-maximal assignment, which is possible in polynomial time. 
A set inclusion between the two rules would be interesting, as then \uc would be a natural rule that is relatively decisive and yet contains all rank-maximal \emph{and} all popular assignments.
\end{remark}

\begin{remark}
Generous assignments \citep{MeMi06a,Manl13a} 
are a dual version of rank maximal assignments: we compare assignments again based on their rank vector, but we now lexicographically optimize for the worst-off agents. While generous and rank-maximal assignments have very similar definitions, \uc indicates that rank-maximal assignments may be preferable, as there are simple instances where no generous assignment is in the \uc. For example, consider the following profile $P$.
\[
\arraycolsep=3pt P = \begin{array}{rccccccc}
1\colon & \goodass{\underline c} & f & a & e & b & g & d \\
2\colon & \goodass{\underline b} & c & g & e & a & d & f \\
3\colon & \goodass g & f & a & \underline d & e & c & b \\
4\colon & {\underline g} & b & e & c & a & f & \goodass d \\
5\colon & \goodass {\underline e} & d & a & b & c & f & g \\
6\colon & \goodass {\underline a} & b & d & g & f & e & c \\
7\colon & \goodass {\underline f} & b & d & e & c & a & g \\
\end{array}
\]
The only two uncovered assignments for $P$ are the underlined $\mu = (c,b,d,g, e,a,f)$ and the blue $\lambda = (c,b,g,d, e,a,f)$. However, neither of these is generous, as we can modify $\mu$ by giving agent $3$ house $a$ and agent $6$ house $d$. This modified assignment $\mu'$ gives all agents houses within their top three, while both $\mu$ and $\lambda$ give some agent her fourth-best house or worse. Hence, no generous assignment in this profile is uncovered.
\end{remark}

\section{Conclusion and Future Work}

In this paper, we initiate a systematic study of majoritarian assignment rules---set-valued assignment rules that rely solely on the pairwise majority relation. Prior work on majoritarian concepts in the context of assignment was restricted to popularity, corresponding to weak Condorcet winners in social choice theory.\footnote{\citet{KaVa23a} have studied semi-popularity and Copeland winners in the more general setting of roommate markets.} 
However, just like weak Condorcet winners, popular assignments rarely exist. To circumvent this issue, social choice theory has developed a range of majoritarian functions that return sets of ``good'' alternatives in the absence of Condorcet winners. We have transferred two of the most prominent such functions---the top cycle and the uncovered set---to the subdomain of assignment. These rules are symmetric, treating all agents and houses equally, and they help to narrow down the set of acceptable assignments, from which a final selection (e.g., by randomization) can be made.

We proved a structural result about assignment-induced majority graphs, which, surprisingly, revealed that some well-known assignment rules are majoritarian. 
We then gave a complete and efficiently checkable characterization of the assignments contained in the top cycle. This characterization reveals that the top cycle not only contains all Pareto-optimal assignments (which does not hold in the more general social choice domain) but also some rather unattractive ones. The top cycle is too coarse to exclude these undesirable assignments. By contrast, the three variants of the uncovered set we studied are much more selective. In fact, each of them contains a subset of all Pareto-optimal assignments and thus offers a promising foundation for new, appealing assignment rules. 

Our findings pave the way for the exploration of further appealing refinements of the McKelvey uncovered set, such as the minimal covering set and the bipartisan set (aka sign-essential set) \citep[see, e.g.,][]{DuLa99a,BrFi08b,BBH15a}. Both of these rules can be computed efficiently in social choice theory. Whether this is also true in assignment is wide open. 
Indeed, even seemingly simpler problems---such as finding an assignment in the uncovered set or the Copeland set, or deciding whether a given assignment belongs to any of these sets---remain unresolved. 
In particular, it would be interesting to investigate whether---in contrast to social choice---the support of mixed popular assignments is contained in the uncovered set.

Other avenues for future research include relaxations of the model that allow for different numbers of agents and houses, ties in the preferences, and pairwise matchings of agents. These generalizations would broaden the applicability of majoritarian assignment rules and deepen our understanding of their properties.

\begin{acks}
This work was supported by the Deutsche Forschungsgemeinschaft under grant BR 2312/11-2. Most of this research was carried out while Patrick Lederer was at UNSW Sydney, where he was funded by the NSF-CSIRO project on ``Fair Sequential Collective Decision Making'' (RG230833). Similarly, Chris Dong was at TU Munich, supported by the above grant of the Deutsche Forschungsgemeinschaft.
We greatly appreciate the contributions of Martin Bullinger and Matthias Greger during our initial research meetings on this topic. Parts of this work were presented at the rump session of the PRAGMA Fest 2025 in Krakow.
\end{acks}

\clearpage

\appendix
\section{Proof of Theorem \ref{thm:C1_Char}}\label{app:thm1}
In this appendix, we present the proof of \Cref{thm:C1_Char}, which we break down into the following lemmas. \begin{lemma}\label{lem:MajorityQuery_pq}
    Let $x,y\in N$ and $p,q\in H$ be pairwise distinct. 
    Consider any $\mu$ with $\mu(x)= p$, $\mu(y)=q$. Consider the matching $\lambda$ which only differs from $\mu$ in $\lambda (x)=q$ and $\lambda (y)=p$. Then, the following holds:
    \begin{itemize}
        \item $\mu \mpref \lambda$ iff $p\pref_x q$ and $q \pref_y p$,
        \item $\mu \dispref \lambda$ iff $q\pref_x p$ and $p \pref_y q$, and
        \item $\mu \sim \lambda$ iff $x,y$ have the same preferences over $p,q$ (i.e., either $p\pref_x q$ and $p \pref_y q$ or $q\pref_x p$ and $q \pref_y p$).
    \end{itemize}
\end{lemma}
\begin{proof}
    Since $\mu(z)= \lambda(z)$ for all agents $z\in N\setminus\{ x,y\}$, the majority comparison only depends on the preferences of $x,y$ between $p,q$. Going through all three cases, we see that the statement clearly holds.
    As aid, we visualize the assignments $\goodass \mu$ and $\badass \lambda$ in blue and red, respectively, for all three cases.
    \begin{itemize}
    \item $x$ prefers $p$ to $q$ but $y$ does not (hence $\mu \pref \lambda$):
    \[\arraycolsep=3pt P = \begin{array}{rcc}
    x\colon & \goodass p,& \badass q\\
	y\colon & \goodass q,& \badass p\\
	\end{array}
    \]
    \item $x$ prefers $q$ to $p$, but $y$ does not (hence $\mu \dispref \lambda$):
    \[\arraycolsep=3pt P' = \begin{array}{rcc}
    x\colon & \badass q, & \goodass p\\
	y\colon & \badass p, & \goodass q\\
	\end{array}
    \]
    \item $x,y$ have the same preferences over $p,q$ (hence $\mu \sim \lambda$):
        \[\arraycolsep=3pt P'' = \begin{array}{rcc}
    x\colon & \badass q, & \goodass p\\
	y\colon & \goodass q,&\badass p\\
	\end{array}\qquad
    P''' = \begin{array}{rcc}
    x\colon & \goodass p,& \badass q\\
	y\colon & \badass p, &\goodass q\\
	\end{array}
    \]
    \end{itemize}
    
\end{proof}

Fix any $p,q\in H$. We apply \Cref{lem:MajorityQuery_pq} for the agent pair $1,2$, then $2,3$, and so on. Either all agents have the same preferences over the pair $p, q$, or for some pair $x,x+1$ we can determine the preferences of these two agents over $p,q$ and hence obtain the pairwise preferences of all agents. 

We build a graph with the houses as nodes and edges between two houses if we were able to fully determine the agents' preferences between these houses. Let us call the connected components of this graph $H_1,\dots, H_k$.
Next, we focus on each connected component $H_i$. We can fully determine the preferences within each component by virtue of the following lemma.
\begin{lemma}\label{lem:pathstep}
    Let $p,q,r\in H$ be distinct such that for all agents, we know their preferences between $p$ and $q$, as well as between $q$ and $r$. Then, we also know all agents' preferences between $p$ and $r$.
\end{lemma}
\begin{proof}
    If not all agents rank $p$ over $r$ or not all agents rank $r$ over $p$, we can apply \Cref{lem:MajorityQuery_pq} to obtain all preferences of all agents between $p$ and $r$. Hence, let all agents rank $p$ over $r$, or let all agents rank $r$ over $p$.
    It now suffices to determine the preferences of a single agent.
    If there is an agent $x$ such that $p\pref_x q \pref_x r$ or $r\pref_x q\pref_x p$, we know by transitivity that $p\pref_x r$ or $r \pref_x p$, respectively. Otherwise, all agents $x$ rank $q$ either first or last among $\{q,p,r\}$ (i.e., either $q \pref_x p,r$ or $p,r \pref_x q$). Because the case $n \le 2$ is trivial, consider $n \ge 3$. By the pigeonhole principle, there are at least two agents $x$ and $y$ of the same type. We only consider the case where $q \pref_x p,r$ and $q \pref_y p,r$ since the other case is symmetric. Now, consider any third agent $z\in N$ and examine the two assignments $\goodass\mu$ and $\badass\lambda$, which we will first visualize.

    \[\arraycolsep=3pt P = \begin{array}{rccc}
    x\colon & \goodass q, & p,\badass r\\
	y\colon & \badass q,&\goodass p, r\\
	z\colon & \dots    &\badass{p}, \goodass{r}\\
    \end{array}
    \]

    Formally, we take any $\mu\in M$ with $\mu(x) = q, \mu(y) = p$, and $\mu(z) = r$. We obtain $\lambda$ by altering $\mu$ in $\lambda(x) = r, \lambda(y) = q$, and $\lambda(z) = p$.
    Since $\mu(w) = \lambda(w)$ for all $w\in N\setminus \{ x,y,z\}$, the majority comparison only depends on the agents $x,y,z$. Clearly, $x$ prefers $\mu $ to $\lambda$, and $y$ prefers $\lambda$ to $\mu$. Hence, there are two possibilities. If $\mu \pref \lambda$, then $r \pref_z p$. Otherwise $\mu \dispref \lambda$, due to $r \dispref_z p$. We have thus determined the preferences of agent $z$ regarding houses $p$ and $r$. This concludes the proof.
\end{proof}

Based on the last two lemmas, we can now identify sets of houses $H_1,\dots, H_k$ for which we can fully specify the agents' preferences. To this end, consider the following graph $G$, whose vertices are the houses $H$, and there is an edge between two houses $p$ and $q$ if applying \Cref{lem:MajorityQuery_pq} for all agent pairs allows us to infer the agents' preferences over $p$ and $q$. Then, the sets $H_1,\dots,H_k$ correspond to the connected components of this graph. In particular, for any two houses $p,q$ that are connected in this graph, there is a sequence of houses $p_1=p,\dots, p_\ell=q$ such that $p_i$ is a neighbor of $p_{i+1}$ in $G$. By repeatedly applying \Cref{lem:pathstep} along this path, we can infer the preferences of all agents between $p$ and $q$.

Moreover, we claim that the houses in each set $H_i$ must be ranked contiguously by each agent, i.e., for all $H_i$, agents $x\in N$, and houses $p,q\in H_i$, there is no house $r\not\in H_i$ such that $p\succ_x r\succ_x q$. Assume for contradiction that this is not true, which means that there is a connected component $H_i$ of $G$, an agent $x$, and houses $p,q\in H_i$, $r\not\in H_i$ such that $p\succ_x r\succ_x q$. Further, we partition $H_i$ into the sets $H_i^+=\{p'\in H_i\colon p'\succ_x r\}$ and $H_i^-=\{q'\in H_i\colon r\succ_x q'\}$. Since $r\not\in H_i$, we know that there is no edge from $r$ to any house $h\in H_i$ in $G$. This means all agents agree on the preference between $r$ and the houses in $H_i$, i.e., it holds for all $y\in N$ that $p'\succ_y r$ for all $p'\in H_i^+$ and $r\succ_y q'$ for all $q'\in H_i^-$. However, by the transitivity of the agents' preferences, this means that $p'\succ_y q'$ for all $y\in N$, $p'\in H_i^+$, and $q'\in H_i^-$. In turn, this implies that there is no edge between houses in $H_i^+$ and $H_i^-$ in $G$, which contradicts that $H_i=H_i^+\cup H_i^-$ is a connected component in this graph. 

We note that it follows from the last paragraph also that, for all components $H_i$, $H_j$, it either holds that $p\succ_x q$ for all $x\in N$, $q\in H_i$, $p\in H_j$, or $q\succ_x p$ for all $x\in N$, $q\in H_i$, $p\in H_j$. The reason for this is that all agents rank the houses in $H_i$ and $H_j$ contiguously. Moreover, if there were two agents $x$ and $y$ such that $p\succ_x q$ and $q\succ_yp$ for all $p\in H_i$, $q\in H_j$, then $H_i$ and $H_j$ would be placed in the same connected component. From this, we infer that the sets $H_1,\dots, H_k$ can be ordered to form a valid decomposition of $P$.

By our analysis so far, we get that every profile that induces $G_{P^*}$ can only differ in the order of the components $H_1,\dots, H_k$ in the agents' preferences. To complete the proof, we need to show that only cyclic shifts of the components result in the same majority graph. This last step can be inferred analogously to Lemma 4 of \citet{BHS16a}. The idea is that, for each three distinct houses $p,q,r$, a linear preference relation either agrees with two out of three of the following comparisons $p\succ_{cyc} q, q\succ_{cyc} r, r\succ_{cyc} p$, or it agrees with two out of three of the following comparisons $p\succ'_{cyc} r, r\succ'_{cyc} q, q\succ'_{cyc} p$. We can use this as follows:
\begin{lemma}\label{lem:cycle_type}
    Let $p,q,r$ be houses such that agents $1,2,3$ satisfy ${\pref_1\vert_{p,q,r }}= {\pref_2\vert_{p,q,r}} = {\pref_3\vert_{p,q,r}}$. Then, from the majority graph, we can determine the ``direction'' in which these three are ordered.
\end{lemma}
\begin{proof}
    Consider any $\mu$ with $\mu(1)= p$, $\mu(2)= q$, $\mu(3)= r$, and $\lambda$ which differs from $\mu$ only in $\lambda(1)= q$, $\lambda(2)= r$, $\lambda(3)= p$. Then, it holds that $\mu\mpref \lambda$ or $\lambda \mpref \mu$, depending on whether the preferences are of type $p\succ'_{cyc} r, r\succ'_{cyc} q, q\succ'_{cyc} p$ or $p\succ_{cyc} q, q\succ_{cyc} r, r\succ_{cyc} p$, respectively.
\end{proof}

For each three distinct components, we can take $p\in H_i$, $q\in H_j$, $r\in H_\ell$ and apply \Cref{lem:cycle_type}. Since each component is contiguous in all preference relations, this gives us the ``cycle type'' of how the three components are ordered in the entire profile. Now, one can prove that any cyclic permutation of the components does not change the ``cycle type'' of any three distinct components. However, any non-cyclic permutation of the components yields three distinct components $H_i,H_j,H_r$ such that their cycle type is changed. By choosing $p\in H_i, q\in H_j,r\in H_r$, we can use the assignment in \Cref{lem:cycle_type} to obtain $\mu$ and $\lambda $ for which the majority comparison does not coincide with the respective one in $G_{P^*}$.

\section{Proof of Theorem \ref{thm:Structure_TC}}\label{app:TC}
In this section, we present the proof of \Cref{thm:Structure_TC} in two steps. First, we introduce Pareto-pessimality and the bottom cycle, which are dual concepts to Pareto-optimality and the top cycle. 
Then, we prove that all Pareto-optimal assignments are contained in the top cycle and characterize the cases in which the top cycle is of size at most two. We obtain analogous statements for the bottom cycle.
Then, in step two, we show that the top cycle and bottom cycle coincide whenever they are not of size one or two. This requires a proof by induction. We prove the base case for five agents by computer, whereas the induction step is shown by hand.

For the proof of \Cref{thm:Structure_TC}, we have to consider concepts for assignments that are particularly ``bad''.
Given a profile $P$, by $P^{-1}$, we denote the profile where all agent preferences $\pref_x$ are inverted to $\pref^{-1}_x$, i.e., $p\pref_x q$ if and only if $q \pref^{-1}_x p$. The majority relation induced by $P^{-1}$ is precisely the inverse of the majority relation $\mprefsim$ of $P$, and we hence denote it by $\mprefsim^{-1}$.

As a dual concept to a serial dictatorship, a serial antidictatorship works exactly the other way around.
The agents pick their \emph{least} preferred houses that are still available in order of $\order$.

Analogously to the top cycle, the bottom cycle consists of all assignments that are minimal elements in the transitive closure of the weak dominance relation.
Formally, $\bc(P)=\left\{\mu\in M_{N,H}\colon\forall\nu:\nu\mprefsimt\mu\right\}$. An assignment is \textit{Pareto-pessimal} if it does not Pareto-dominate any other assignment.
By $\pso{}(P) = \po(P^{-1})$, we denote the set of all Pareto-pessimal assignments in $P$. It holds that an assignment is Pareto-pessimal if and only if it can be obtained as a serial antidictatorship.

\begin{lemma}\label{lem:dualTCBC}
     For all profiles $P$, we have $\tc(P) = \bc(P^{-1})$ and  $\bc (P)=\bc ({(P^{-1})}^{-1}) = \tc(P^{-1})$.
\end{lemma}
\begin{proof}
    	Note that the majority relation over the assignments induced by $P$ is precisely inverse to the majority relation induced by $P^{-1}$, and hence so is the transitive closure over it.
Thus, $\mu\mprefsimt\lambda$ for all assignments $\lambda$ if and only if $\lambda\mathrel{{\left(\mprefsim^{-1}\right)}^*}\mu$ for all assignments $\lambda$.
    	This proves the claim.
\end{proof}

For example, \Cref{lem:dualTCBC} implies that each Pareto-pessimal assignment belongs to the bottom cycle, as Pareto-pessimal assignments of a profile $P$ are Pareto-optimal assignments of $P^{-1}$.

We first characterize all profiles in which  \tc chooses at most two assignments. For this, we first prove \Cref{thm:PO_Subset_TC}. Further, by \Cref{lem:dualTCBC}, this means that $\pso\subseteq\bc$.

\POSubseteqTC*
\begin{proof}
We will prove that, for every profile $P$, each Pareto-optimal assignment can reach each other Pareto-optimal assignment via a path in the majority graph. This shows that $\po(P)\subseteq \tc(P)$ because each assignment that is not Pareto-optimal is Pareto-dominated (and thus also majority dominated) by an assignment in $\po(P)$. To prove this claim, we will rely on the characterization of Pareto-optimal assignments by \citet{AbSo98a}, which states that an assignment $\mu$ is in $\po(P)$ if and only if there is a priority order $\order$ over the agents such that $\mu$ is the outcome of the serial dictatorship for $\sigma$.

    Before showing that all Pareto-optimal assignments are connected by paths in the majority graph, we prove an auxiliary statement: when modifying a picking sequence $\sigma$ by improving the position of the last agent, the outcome of the serial dictatorship for the original sequence weakly majority dominates the outcome of the serial dictatorship for the modified sequence. To make this more formal, fix an order $\sigma=(x_1,\dots, x_n)$ over the agents and let $\sigma'=(x_1,\dots,x_{k-1},x_n,x_{k},\dots,x_{n-1})$ for some $k\in \{1,\dots,n-1\}$. Moreover, let $\mu$ and $\lambda$ denote the assignment picked by the serial dictatorship induced by $\sigma$ and $\sigma'$, respectively. We will show that $\mu\mprefsim \lambda$. 
    
   To prove this, let $X_i$ and $X_i'$ denote the houses that are available when agent $x_i$ gets to pick her house under $\sigma$ and $\sigma'$, respectively. We claim that $X_i\supseteq X_i'$ for all $i\in \{1,\dots, n-1\}$. First, for the agents $x_i\in \{x_1,\dots, x_{k-1}\}$, it holds even that $X_i=X_i'$ because $\order$ and $\order'$ agree on the first $k-1$ agents. Next, we have that $X_{k}=X_{k}'\cup \{\lambda(x_n)\}$ because agent $x_{k}$ gets to pick before $x_n$ in $\sigma$. Now, inductively assume that $X_i\supseteq X_{i}'$ for some $i\in \{k,\dots, n-2\}$. We will show that $X_{i+1}\supseteq X_{i+1}'$. For this, we note that $|X_i|=|X_i'|+1$ because the agents $x_1,\dots, x_{i-1}$ pick before $x_i$ under $\sigma$, whereas $x_n$ additionally gets to choose before $x_i$ under $\sigma'$. By our induction assumption, we thus conclude that there is a single house $p$ such that $X_i=X_i'\cup \{p\}$. Since $\lambda(x_i)$ is agent $x_i$'s favorite house in $X_i'$, this means that agent $x_i$ either picks $p$ or $\lambda(x_i)$ from $X_i$. Consequently, $X_{i+1}=X_i'$ if agent $x_i$ picks $p$, or $X_{i+1}=X_{i+1}'\cup\{p\}$ if she picks $\lambda(x_i)$. In both cases, it holds that $X_{i+1}\supseteq X_{i+1}'$, thus proving the induction step. 

    By the definition of serial dictatorships, the fact that $X_i\supseteq X_i'$ for all $x_i\in \{x_1,\dots, x_{n-1}\}$ implies that $\mu\succeq_{x_i}\lambda$ for all these agents. If one of these agents strictly prefers $\mu$ to $\lambda$, we have that $\mu \mprefsim \lambda$ as only agent $x_n$ may prefer $\lambda$ to $\mu$. On the other hand, if no $x_i\in \{x_1,\dots, x_{n-1}\}$ strictly prefers $\mu$ to $\lambda$, then $\mu(x_i)=\lambda(x_i)$ for all these agents, which implies that $\mu=\lambda$ and thus again $\mu \mprefsim \lambda$.

   Based on our auxiliary claim, we will now complete the proof of this lemma. To this end, we fix two distinct assignments $\mu,\lambda\in\po(P)$ and consider two orders $\order=(x_1,\dots, x_n)$ and $\order'=(x_1',\dots,x_n')$ such that the corresponding serial dictatorship choose $\mu$ and $\lambda$ for $P$. We will iteratively transform $\order$ into $\order'$. To this end, let $i$ denote the smallest index such that $x_i\neq x_i'$ and let $j>i$ denote the index such that $x_j=x_i'$. Now, consider the priority order $\sigma_1$ derived from $\sigma$ by placing the currently last-ranked agent directly before $x_j$, and let $\eta_1$ be the assignment chosen by the corresponding serial dictatorship. By our auxiliary claim, we have $\mu\mprefsim\eta_1$. Further, we can repeat this step until we have a priority order $\sigma_2$ where $x_{i}'=x_j$ is in the last position, and our auxiliary claim shows that $\mu\mprefsimt \eta_2$ for the assignment $\eta_2$ chosen by the corresponding assignment. Next, let $\sigma_3$ be the priority order derived from $\sigma_2$ by moving $x_{i'}$ directly before $x_i$. Again using our auxiliary claim, the assignment $\eta_3$ chosen by the corresponding serial dictatorship satisfies that $\eta_2\mprefsim\eta_3$, so we have $\mu\mprefsimt \eta_3$. Lastly, we observe that the priority orders $\sigma_3$ and $\sigma'$ agree on the first $i$ positions. Hence, by repeating our argument at most $n$ times, we transform $\sigma$ into $\sigma'$ while constructing a path from $\mu$ to $\lambda$ in the majority graph. 
\end{proof}

We next use \Cref{thm:PO_Subset_TC} to analyze cases \textit{(i)} and \textit{(ii)} of \Cref{thm:Structure_TC}, as well as analogous statements regarding the bottom cycle. First, we will study profiles where all agents have distinct top choices (or bottom choices).
\begin{lemma}\label{lem:TC_Cardinality_1}
	$|\tc(P)|=1$ if and only if all agents have distinct top choices.
Analogously, $|\bc(P)|=1$ if and only if all agents have distinct bottom choices.
\end{lemma}
\begin{proof}
    Let $P$ be a profile where all agents have distinct top choices.
    Let $\mu$ be the assignment that assigns to each agent her top choice.
    For all other assignments $\lambda$, there is at least one agent who does not obtain their top choice.
    These agents strictly prefer $\mu$ to $\lambda$, while all other agents weakly prefer $\mu$ to $\lambda$.
    Thus, we have $\mu \pref \lambda$ for all assignments $\lambda\neq \mu$, which implies $\tc(P) = \{\mu\}$.

    We prove the other implication for $\tc$ by contraposition.
    Let $P$ be a profile where two agents $x\neq y$ have the same favorite house $p$.
    By \Cref{thm:PO_Subset_TC}, we know that all assignments obtained via serial dictatorships are contained in the top cycle.
    Since the serial dictatorship where $x$ chooses first and $y$ second yields a different assignment than the serial dictatorship where $y$ chooses first and then $x$ chooses, there are at least two assignments in the top cycle.
    This proves the statement for the top cycle.

    For the bottom cycle, the statement follows from \Cref{lem:dualTCBC}.
\end{proof}

The first part of \Cref{lem:TC_Cardinality_1} follows from existing literature (e.g.,~\citealp[Proposition 7.24.]{Manl13a}).
For the sake of completeness, we provided a proof nevertheless. We further note that 
\Cref{lem:TC_Cardinality_1} directly proves case \textit{(i)} of \Cref{thm:Structure_TC} and will be useful for the proof of case \textit{(iv)}.

Next, we describe all cases in which the top cycle or the bottom cycle has cardinality two.
\begin{lemma}\label{lem:TC_Cardinality_2}
	$|\tc(P)|=2$ if and only if all but two agents have distinct top choices and the two agents who have the same top choice also share the same second-best choice, which is not top-ranked by any~agent.

	Analogously, $|\bc(P)|=2$ if and only if all but two agents have distinct bottom choices and the two agents who have the same bottom choice also share the same second worst choice, which is not bottom-ranked by any agent.
\end{lemma}
\begin{proof}
	Let $P$ be a profile where two distinct agents $x^*, y^*$ share the same top choice $p$ and second choice $q$, and all other agents have unique favorite houses in $H\setminus \{p,q\}$.
	Consider the two assignments $\mu_{x^*}$, where all agents but $x^*$ obtain their favorite house and $x^*$ obtains $q$, and $\mu_{y^*}$, where all agents but $y^*$ obtain their favorite houses and $y^*$ obtains $q$.
	Clearly, $\mu_{x^*} \sim \mu_{y^*}$.
    Further, we note that in every other assignment $\lambda\in M\setminus \{\mu_{x^*},\mu_{y^*}\}$, at least one agent among $x^*$ and $y^*$ cannot obtain their favorite house.
	Without loss of generality, let that agent be $x^*$.
	Then, every agent weakly prefers $\mu_{x^*}$ to $\lambda$.
    Moreover, since $\lambda\not\in\{\mu_{x^*}, \mu_{y^*}\}$, there is at least one agent $x\in N\setminus \{x^*, y^*\}$ with $\lambda(x)\neq \mu_{x^*}(x)$. Indeed, if $\lambda$ would agree with $\mu_{x^*}$ on all agents in $N\setminus \{x^*, y^*\}$, then either $x^*$ gets $p$ and $y^*$ gets $q$ and $\lambda=\mu_{y^*}$, or $x^*$ gets $q$ and $y^*$ gets $p$ and $\lambda=\mu_{x^*}$.
	This agent $x$ strictly prefers $\mu_{x^*}$ to $\lambda$, so  $\mu_{x^*}\mpref\lambda$.
	Further, we compare $\lambda$ to $\mu_{y^*}$.
    If $\lambda(y^*)\neq p$, we have that $\mu_{y^*}\succeq_x \lambda$ for all $x\in N$, so we immediately get that $\mu_{y^*}\mpref \lambda$. On the other hand, if $\lambda(y^*)= p$, then $\mu_{y^*}(x^*)=p\succ_{x^*} \lambda(x^*)$. Further, since there is another agent $x\not \in \{x^*, y^*\}$ with $\mu_{y^*}\succ_x \lambda$, there are at least two agents strictly preferring $\mu_{y^*}$ to $\lambda$, proving that $\mu_{y^*} \mpref \lambda$.
	This concludes the proof that $\tc(P)$ has cardinality two.

    For the other direction of the $\tc$ statement, let $P$ not be of the above form.
    Then, we distinguish between four cases.
    \begin{enumerate}[label=(\roman*)]
        \item The top choices of the agents are all pairwise different.
        \item There are three agents who share the same top choice.
        \item All agents but $x^*$ and $y^*$ have disjoint top choices, but agents $ x^*$ and $y^*$ have different second choices.
        \item There are two pairs of agents $x,y$ and $z,w$ with coinciding top choices.
    \end{enumerate}
    In Case (i), \Cref{lem:TC_Cardinality_1} directly implies that the cardinality of $\tc$ is not equal to two.
    In Case (ii), we consider three serial dictatorships, where one of these agents chooses their house first, respectively.
    Since the three serial dictatorships yield different outcomes, we obtain by \Cref{thm:PO_Subset_TC} that the top cycle contains at least three assignments.
    In Case (iii), we again consider three serial dictatorships.
    By the pigeonhole principle, the second choice of one of the two agents $x^*$, $y^*$, coincides with the top choice of some agent $z$.
    Without loss of generality, we assume that the favorite house of agent $z$ is the second-ranked house of agent $y^*$.
    In the first serial dictatorship, $y^*$ chooses before $x^*$ before $z^*$.
    In the second, the roles of $x^*$ and $y^*$ are reversed.
    In the third, $z$ chooses first, $x^*$ second, and $y^*$ third.
    The resulting three assignments do not coincide as agent $y^*$ obtains her first-ranked (resp. second-ranked or third-ranked) house in the first (resp. second or third) serial dictatorship. 
    Hence, in Case (iii), there are at least three assignments in the top cycle by \Cref{thm:PO_Subset_TC}.
    Finally, in Case (iv), we can consider all serial dictatorships that vary the picking order between agents $x$, $y$, $z$, and $w$.
    This results in at least four assignments that are serial dictatorships and are, therefore, contained in the top cycle.
    This concludes the proof of the top cycle statement.

    Note that the bottom cycle statement follows from \Cref{lem:dualTCBC} and the statement about the top cycle.
\end{proof}

\Cref{lem:TC_Cardinality_2} directly proves case \textit{(ii)} of \Cref{thm:Structure_TC}. It will further be helpful in proving case \textit{(iii)}.

Lastly, we turn to the cases where $|\tc(P)|>2$. To this end, we first analyze the setting with $n=5$ agents and houses in more detail. In particular, since there is a finite number of profiles when $n=5$, we can simply iterate over all profiles with a computer to infer the following claim. 

\begin{fact}\label{fact:n5}
	For any profile $P$ with $n = 5$, the following statements hold.
\begin{enumerate}
        \item $\lvert \tc(P) \rvert > 2 \iff M
\setminus\pso(P)\subseteq\tc(P)$.
        \item $\lvert \bc(P) \rvert > 2
\iff M
\setminus\po(P)\subseteq\bc(P)$.
        \item $\lvert \tc(P) \rvert > 2$ and $\lvert \bc(P) \rvert > 2 \iff \tc(P) = M
$
        \item If $\lvert \tc(P) \rvert \le 2$ and $ \lvert \bc(P) \rvert  \le 2 $, then for all $\mu, \lambda\in M
\setminus (\tc(P)\cup \bc(P))$ we have $\mu \mprefsimt \lambda$.
    \end{enumerate}
\end{fact}
\begin{proof}
	For $n=5$, the number of possible instances is manageable: up to symmetry, there are about nine million preference profiles.
We are able to verify \Cref{fact:n5} by brute forcing through all possible profiles (up to symmetry) with the help of a household computer in two days. The code can be found on Zenodo \citep{Schl26a}.
\end{proof}

\Cref{fact:n5} implies a very helpful lemma for general $n$, because we can always restrict attention to $5$ agents and their houses. In more detail, in the subsequent proofs, we will often focus on restricted sets of agents and introduce some notation to this end. First, recall that every set of agents $N$ and every set of houses $H$ with $|N|=|H|$ induce a corresponding set of assignments $M$. Now, let $M_{N',H'}$ denote the set of all assignments on the agent set $N'\subseteq N$ and house set $H'\subseteq H$. Some assignments $\mu\in M$ also obtain a corresponding restriction $\mu_{N',H'}$, with domain $N'$ and image $H'$. More formally, $\mu_{N',H'}$ is a function from $N'$ to $H'$ that is derived from $\mu$ by setting $\mu_{N',H'}(x)=\mu(x)$ for all $x\in N'$. Note that we only use this notation if $|N'|=|H'|$ and $\mu(N')=H'$, so the restricted assignment is well-defined on $N'$ and $H'$.
Similarly, for a profile $P$, agent set $N$, and house set $H$, let $P_{N',H'}$ denote the profile restricted to agent set $N'$ and house set $H'$.
The restricted profile $P_{N',H'}$ induces a majority relation on $M_{N',H'}$. This relation is denoted via $\mprefsim_{N',H'}$. Note that we will only use this notation and compare two assignments with respect to $ \mprefsim_{N',H'}$ if they both belong to $M_{N',H'}$.

Using this notation, we next present a lemma that demonstrates how we can use \Cref{fact:n5} for larger values of $n$.

\begin{lemma}\label{lem:non-isolation}
    Let $\mu,\lambda\in M$ such that for some $N'$ of size $5$, we have $\mu(N') = \lambda(N') =: H'$. 
	If $\mu_{N',H'}\not\in \tc(P_{N',H'})$ or $|\tc(P_{N',H'})|>2$, and $\lambda_{N',H'}\not\in \bc(P_{N',H'})$ or $|\bc(P_{N',H'})|>2$, then $\lambda_{N',H'} \mprefsimt_{N',H'}  \mu_{N',H'}$. Moreover, if $\mu(x)=\lambda(x)$ for all $x\in N\setminus N'$, we have $ \lambda\mprefsimt \mu$.
\end{lemma}
\begin{proof}
    Fix a set of agents $N'$, a set of houses $H'$, and two assignments $\mu$ and $\lambda$ that satisfy the conditions of the lemma.
    First, if $|\tc(P_{N',H'})|>2$ and $|\bc(P_{N',H'})|>2$, then $\tc(P_{N',H'})=M_{N',H'}$ by Claim (3) of \Cref{fact:n5} and $\lambda_{N',H'} \mprefsimt_{N',H'} \mu_{N',H'}$ because each assignment in the top cycle has a path to every other assignment. Second, if $|\tc(P_{N',H'})|\leq 2$ and $|\bc(P_{N',H'})|>2$, it follows that $\tc(P_{N',H'})\neq \bc(P_{N',H'})$. Further, we know that $\po(P_{N',H'})\subseteq \tc(P_{N',H'})$ by \Cref{thm:PO_Subset_TC} and that $M_{N',H'}\setminus \po(P_{N',H'})\subseteq \bc(P_{N',H'})$ by Claim (2) of \Cref{fact:n5}. Because $\mu_{N',H'}\not\in\tc(P_{N',H'})$ if $|\tc(P_{N',H'})|\leq 2$, this means that $\mu_{N',H'}\in \bc(P_{N',H'})$ and thus $\lambda_{N',H'} \mprefsimt_{N',H'}  \mu_{N',H'}$. Thirdly, assume that $|\tc(P_{N',H'})|> 2$ and $|\bc(P_{N',H'})|\leq 2$. In this case, we get that $M_{N',H'}\setminus \pso(P_{N',H'})\subseteq \tc(P_{N',H'})$ by Claim (1) of \Cref{fact:n5} and $\pso(P_{N',H'})\subseteq \bc(P_{N',H'})$ by the dual of \Cref{thm:PO_Subset_TC}. Since $\lambda_{N',H'}\not\in\bc(P_{N',H'})$ by assumption, this means that $\lambda_{N',H'}\in \tc(P_{N',H'})$ and our statement follows. Lastly, if both $|\tc(P_{N',H'})|\leq 2$ and $|\bc(P_{N',H'})|\leq2$, we have that $\mu_{N',H'}\in M_{N',H'}\setminus \tc(P_{N',H'})$ and $\lambda_{N',H'}\in M_{N',H'}\setminus \bc(P_{N',H'})$. If $\mu_{N',H'}\in \bc(P_{N',H'})$, $\lambda_{N',H'} \mprefsimt_{N',H'}  \mu_{N',H'}$ follows by definition of the bottom cycle. If $\lambda_{N',H'}\in \tc(P_{N',H'})$, the statement follows by definition of the top cycle. Finally, if $\mu_{N',H'}\in M_{N',H'}\setminus (\tc(P_{N',H'})\cup \bc(P_{N',H'}))$ and $\lambda_{N',H'}\in M_{N',H'}\setminus (\tc(P_{N',H'})\cup \bc(P_{N',H'}))$, the statement follows by Claim (4) of \Cref{fact:n5}.

    In order to extend the argument to the unrestricted assignments $\mu$ and $\lambda$, it suffices to note that the agents in $N\setminus N'$ do not affect the majority comparison between any pair of assignments that only differ in the houses of the agents in $N'$.
\end{proof}

Using \Cref{fact:n5} and \Cref{lem:non-isolation}, we analyze in several steps the outcome of the top cycle when $|\tc(P)|>2$. First, we show that if $|\tc(P)|>2$ for some profile $P$, then there is an agent $x$ and a house $p$ such that $x$ prefers $p$ the most and every assignment in which $x$ gets $p$ is in $\tc(P)$.

\begin{lemma}\label{lem:onetop}
    For every profile $P$ with $n\geq 5$ and $|\tc(P)|>2$, there is an agent $x$ and a house $p$ such that $x$ prefers $p$ the most and every assignment $\mu$ with $\mu(x)=p$ is in $\tc(P)$. 
\end{lemma}
\begin{proof}
    Fix a profile $P$ with $n\geq 5$ agents and $|\tc(P)|>2$. By \Cref{lem:TC_Cardinality_1,lem:TC_Cardinality_2}, this means that $P$ does \emph{not} satisfy that 
    \begin{enumerate}[label=(\roman*)]
        \item all agents have different top choices, or
        \item all but two agents have different top choices, and the remaining two agents agree on their best and second-best house, which are both not top-ranked by any other agent.
    \end{enumerate}
    Conversely, this means that $P$ satisfies (at least) one of the following conditions: 
    \begin{enumerate}[label=(\arabic*)]
    \item There is a house that is top-ranked by at least three agents.
    \item There are at least two houses that are top-ranked by two agents each.
    \item Two agents have the same top choice, but at least one of them second-ranks a house that is top-ranked by a third agent.
\end{enumerate}
    Clearly, none of these three cases is compatible with the cases where $|\tc(P)|\leq 2$. To see that this list is indeed exhaustive, assume that $P$ does not satisfy any of our criteria. By cases (1) and (2), no house is top-ranked by three or more agents and there is at most one house that is top-ranked by two agents. If no such house exists, we are in case (i) and $|\tc(P)|=1$. Hence, assume that there is a house $p$ that is top-ranked by exactly two agents. Since case (3) does not apply, we know that neither of the two agents $x$ and$ y$ who top-rank $p$ second-ranks a house that is top-ranked by another agent. Because no other house is top-ranked by two agents, this means that both $x$ and $y$ second-rank the single house $q$ that is not top-ranked by any agent. Hence, we are in case (ii) and $|\tc(P)|=2$. 

    Now, to prove our lemma, we will continue with a case distinction regarding our three cases.\medskip

    \textbf{Case (1)}:
    First, assume that there are three agents $x,y,z$ that top-rank the same house $p$. Further, let $\lambda$ denote an arbitrary assignment with $\lambda(x)=p$. We will show that $\lambda\in\tc(P)$. To this end, let $\mu$ denote the assignment obtained from any serial dictatorship where $x$ picks first. By this definition, we have that $\mu(x)=p$ and, moreover, $\mu\in \tc(P)$ because $\mu$ is Pareto-optimal and $\po(P)\subseteq \tc(P)$ by \Cref{thm:PO_Subset_TC}. 
    
    Next, pick any agents $u, v\in N\setminus \{x\}$, $u \neq v$,  such that $\mu(v)=\lambda(u)$. Since $\mu(x)=\lambda(x)$, such agents exist unless $\mu=\lambda$ (in which case we are done as $\mu\in \tc(P)$). 
    Let $\mu'$ denote the assignment obtained from $\mu$ by swapping the houses of agents $u$ and $v$. We aim to show that $\mu'\mprefsimt \mu$ and thus $\mu'\in\tc(P)$. To this end, let $N'=\{x,y,z,u,v\}$ and let $H'=\mu(N')$ denote the houses these agents obtain under $\mu$. Further, if $\{u,v\}\cap \{y,z\}\neq\emptyset$, we fill up $N'$ with arbitrary agents so that $|N'|=5$ and $H'$ with the houses of these agents under $\mu$. Now, let us consider the instance $P_{N',H'}$. Since $p\in H'$ and the agents $x,y,z$ top-rank $p$, we obtain from \Cref{lem:TC_Cardinality_1,lem:TC_Cardinality_2} that $|\tc(P_{N',H'})|>2$. Further, if $|\bc(P_{N',H'})|\leq 2$, then $\mu'\not\in \bc(P_{N',H'})$ since agent $x$ still obtains his favorite house $p$. In particular, note here that \Cref{lem:TC_Cardinality_1,lem:TC_Cardinality_2} show that, if $|\bc(P_{N',H'})|\leq 2$, an assignment belongs to the bottom cycle only if every agent gets at best his second-least preferred house. Due to these insights and since $\mu(w)=\mu'(w)$ for all $w\in N\setminus N'$, we derive that $\mu'\mprefsimt \mu$ by \Cref{lem:non-isolation}. 

    Finally, we note that $\mu'$ differs in the assignment of one less agent from $\lambda$ than $\mu$. Hence, by repeating the above argument, we can construct a sequence of assignments $\mu_0=\mu,\dots, \mu_k=\lambda$ such that $\mu_k\mprefsimt \mu_{k-1}\mprefsimt\dots\mprefsimt \mu_0$, which shows that $\lambda=\mu_k\in\tc(P)$.\medskip
   
    \textbf{Case (3)}: We will next deal with our third case (and postpone Case (2) as it requires a slightly more involved argument). Hence, assume that there are two agents $x,y$ that have the same top choice $p$ and that agent $x$ second-ranks the top choice $q$ of a third agent $z$. Just as before, we let $\lambda$ denote an arbitrary assignment with $\lambda(x)=p$ and aim to show that $\lambda\in \tc(P)$. To this end, we denote by $\lambda'$ the assignment derived from $\lambda$ by letting agent $z$ and the agent $z'$ who has house $q$ in $\lambda$ swap their houses, i.e., it holds that $\lambda'(z)=q$. If $\lambda(z)=q$, we simply set $\lambda=\lambda'$. Further, let $\mu$ denote an assignment obtained by a serial dictatorship where $x$ picks first and $z$ second. This means that $\mu(x)=p$ and $\mu(z)=q$. Moreover, we have that $\mu\in\tc(P)$ because $\mu\in\po(P)$.

    Similar to the last case, we next pick two agents $u,v\in N\setminus \{x,z\}$ such that $\mu(v)=\lambda'(u)$ and let $\mu'$ be the assignment derived form $\mu$ by swapping the houses of these two agents. We will again show that $\mu'\mprefsimt \mu$. To this end, let $N'=\{x,y,z,u,v\}$ and $H'=\mu(N')$; as for Case (1), we can fill up $N'$ and $H'$ with an arbitrary agent and their house if $y\in \{u,v\}$. Now, consider the profile $P_{N',H'}$. We have $p,q\in H'$ since $\mu(x)=p$ and $\mu(z)=q$, so the preferences of $x$, $y$, and $z$ suffice to infer that $|\tc(P_{N',H'})|>2$ by \Cref{lem:TC_Cardinality_1,lem:TC_Cardinality_2}. Further, if $|\bc(P_{N',H'})|\leq 2$, we get from these lemmas that $\mu'\not\in \bc(P_{N',H'})$ as agent $x$ obtains his favorite house. Since $\mu(w)=\mu'(w)$ for all $w\in N\setminus N'$, we can apply \Cref{lem:non-isolation} to derive that $\mu'\mprefsimt \mu$. We can iteratively repeat this construction to move from $\mu$ to $\lambda'$, i.e., there is a sequence of assignments $\mu_0=\mu,\dots, \mu_k=\lambda'$ such that $\mu_k\mprefsimt \mu_{k-1}\mprefsimt\dots\mprefsimt \mu_0$. This proves that $\lambda'\in\tc(P)$. 

    Lastly, to move from $\lambda'$ to $\lambda$, let $N'=\{x,y,z,z',u\}$, where $z'$ is the agent with $\lambda(z')=q$ and $u$ is an arbitrary fifth agent. By definition, $\lambda$ and $\lambda'$ only differ in the houses of $z$ and $z'$. Now, we can again apply \Cref{lem:non-isolation} for $N'$ and the houses $H'=\lambda'(N')$ to infer that $\lambda\mprefsimt \lambda'$. In more detail, since $p,q\in H'$ and $x,y,z\in N'$, it follows that $|\tc(P_{N',H'})|>2$. Further, the fact that agent $x$ obtains his favorite house still entails that $\lambda\not\in \bc(P_{N',H'})$ if $|\bc(P_{N',H'})|\leq 2$. Hence, by the same argument as before, we get that $\lambda\mprefsimt \lambda'$ and thus $\lambda\in\tc(P)$.\medskip

\textbf{Case (2)}: Finally, assume that there are four distinct agents $x,y,z,w$ and two houses $p,q$ such that $x$ and $y$ top-rank $p$, and $w$ and $z$ top-rank $q$. Once again, we let $\lambda$ denote an arbitrary assignment with $\lambda(x)=p$ and show that $\lambda\in\tc(P)$. To this end, we define $\lambda'$ as the assignment obtained from $\lambda$ by letting agent $z$ swap his house with the agent $z'$ who has $q$ in $\lambda$. If $z=z'$, we simply set $\lambda'=\lambda$. Further, we let $\mu$ denote an assignment obtained by a serial dictatorship where agent $x$ picks first and agent $z$ second. Hence, we have that $\mu(x)=p= \lambda'(x)$, $\mu(z)=q = \lambda'(z)$, and $\mu\in \tc(P)$. 

    Just as for the previous cases, let $u,v\in N\setminus \{x,z\}$ denote two arbitrary agents such that $\mu(v)=\lambda'(u)$. We will show that $\mu'\mprefsimt \mu$ for the assignment $\mu'$ where agents $u$ and $v$ swapped their houses. If $\{u,v\}=\{y,w\}$, then let $N'=\{x,y,z,w,v'\}$ for an arbitrary fifth agent $v'$ and let $H'=\mu(N')$. Since $p,q\in H'$, the agents $x,y,z,w$ still top-rank these houses in $P_{N',H'}$, thus showing that $|\tc(P_{N',H'})|>2$. Further, since agent $x$ gets his favorite house under $\mu'$, this assignment cannot be in the bottom cycle if $|\bc(P_{N',H'})|\leq 2$. Hence, we infer from \Cref{lem:non-isolation} that $\mu'\mprefsimt\mu$. 
    
    On the other hand, if $\{u,v\}\neq \{y,w\}$, at least one of $y$ and $w$ is not in $\{u,v\}$. Without loss of generality, we assume that $y\not\in \{u,v\}$ and note that this choice will not matter subsequently. In this case, we will break down the swap between $u$ and $v$ into three steps: first $y$ and $u$ exchange their houses, then $y$ and $v$, and then again $y$ and $u$. To this end, let $\mu_1$, $\mu_2$, and $\mu_3$ denote the assignment obtained after the first, second, and third exchange. We note that $\mu_3=\mu'$. We will show that $\mu_3\mprefsimt\mu_2\mprefsimt\mu_1\mprefsimt\mu_0=\mu$. To this end, fix one such assignment $\mu_i$ (with $i\in \{0,1,2\}$) and let $N'=\{x,y,z,w,y'\}$ where $y'=u$ if $i=0$ or $2$ and $y'=v$ if $i=1$, and let $H'$ be the corresponding houses. Further, if $w=y'$, we may add an arbitrary fifth agent and his house, as in the previous cases. Now, it holds that $|\tc(P_{N',H'})|>2$ because $p,q\in H'$ and $x,y$ still top-rank $p$ whereas $z,w$ still top-rank $q$. Further, since $\mu_{i+1}(x)=p$, we again get that $\mu_{i+1}\not\in \bc(P_{N',H'})$ unless $|\bc(P_{N',H'})|>2$. Hence, $\mu_{i+1}\mprefsimt \mu_i$ for all $i\in \{0,1,2\}$ by \Cref{lem:non-isolation} and thus $\mu'\mprefsimt \mu$.

    Now, just as for the previous cases, we can repeatedly use this construction to prove that $\lambda'\mprefsimt \mu$. To further prove $\lambda \mprefsimt \lambda'$, recall that $z'$ is the agent with $\lambda(z')=q$.
    In the case that $z'=z$, we defined $\lambda' = \lambda$ and there is nothing to show. If $z\neq z'$, we recall that $\lambda$ arises from $\lambda'$ by only swapping the houses of $z$ and $z'$.
    Let $N'=\{x,y,z,w,z'\}$ and $H'=\lambda(N')$. Just as before, we add an arbitrary agent and his house to $N'$ and $H'$ if $z'$ turns out to be among $\{y,w\}$.
    Now, just as for the previous cases, the presence of the agents $x$, $y$, $z$, and $w$ and the houses $p$ and $q$ shows that $|\tc(P_{N',H'})|>2$. Further, $\lambda\not\in\bc(P_{N',H'})$ if $|\bc(P_{N',H'})|\leq 2$ since agent $x$ obtains his favorite house. Hence, \Cref{lem:non-isolation} shows again that $\lambda\mprefsimt\lambda'$ and thus $\lambda\in\tc(P)$.
\end{proof}

Based on \Cref{lem:onetop}, we will next generalize Claims (1) and (2) of \Cref{fact:n5} to arbitrary values of $n\geq 5$.

\begin{lemma}\label{lem:notpso_subseteq_tc}
    For every profile $P$ with $n\geq 5$ and $|\tc(P)|>2$, it holds that $M\setminus \pso(P)\subseteq \tc(P)$. Analogously, if $n\geq 5$ and $|\bc(P)|>2$, it holds that $M\setminus \po(P)\subseteq \bc(P)$.
\end{lemma}
\begin{proof}
    We first note that the claim on $\bc$ follows from \Cref{lem:dualTCBC} once we have proven the claim for $\tc$. Further, we may assume that $n\geq 6$ because \Cref{fact:n5} proves our lemma when $n=5$. Let $P$ be a profile with $n\geq 6$ agents such that $|\tc(P)|>2$. By \Cref{lem:onetop}, there is an agent $x$ and a house $p$ such that $x$ top-ranks $p$ and every assignment $\mu$ with $\mu(x)=p$ is in the top cycle. Next, let $\lambda$ denote an arbitrary assignment that is not Pareto-pessimal. If $\lambda(x)=p$, $\lambda\in\tc(P)$ by our previous insight, and we are done. Hence, we suppose that $\lambda(x)\neq p$ and show that $\lambda\mprefsimt \mu$ for some assignment $\mu$ with $\mu(x)=p$. This proves $\lambda\in \tc(P)$ since $\mu\in \tc(P)$.

    To show our claim, we note that, since $\lambda$ is not Pareto-pessimal, there is a Pareto-pessimal assignment $\eta\neq \lambda$ such that $\lambda(w)\succeq_w \eta(w)$ for all agents $w$. This holds because, in the inverted profile $P^{-1}$, the Pareto-optimal assignments correspond exactly to the Pareto-pessimal ones in $P$. Further, every assignment that is not Pareto-optimal in $P^{-1}$ is Pareto-dominated by a Pareto-optimal assignment, so every assignment that is not Pareto-pessimal in $P$ Pareto-dominates a Pareto-pessimal assignment. Thus, we can set $\eta$ to be the Pareto-pessimal assignment that $\lambda$ Pareto-dominates. 
    
    Now, let $S=\{w\in N\colon \lambda(w)\succ_w\eta(w)\}$ be the set of agents who strictly prefer their house under $\lambda$ to that under $\eta$. Since $\lambda\neq \eta$, and two distinct assignments must differ in at least two assigned houses, we directly obtain from the Pareto-domination that $\lvert S\rvert \ge 2$. Further, let $y$ denote the agent who obtains house $p$ under $\eta$, i.e., $\eta(y)=p$. We note that $y\neq x$; otherwise, $\lambda$ cannot Pareto-dominate $\eta$ since agent $x$ top-ranks $p$ but $\lambda(x)\neq p$. To prove our lemma, we proceed with a case distinction on the size of $S\setminus \{x,y\}$.\medskip

    \textbf{Case 1: $|S\setminus \{x,y\}|\geq 2$.} First, suppose that $|S\setminus \{x,y\}|\geq 2$. In this case, let $\mu$ denote the assignment derived from $\eta$ by letting agents $x$ and agent $y$ swap their assigned house. We note that $\lambda(w)=\eta(w)=\mu(w)$ for all agents $w\in N\setminus (S\cup \{x,y\})$ and $\lambda(w)\succ_w\eta(w)=\mu(w)$ for all $w\in S\setminus \{x,y\}$. Hence, at least two agents in $N\setminus \{x,y\}$ strictly prefer $\lambda$ to $\mu$, and none of these agents prefers $\mu$ to $\lambda$. This implies that $\lambda\mprefsim \mu$, regardless of the preferences of agents $x$ and $y$. Since $x$ obtains his favorite house $p$ in $\mu$, this completes the proof in this case.\medskip

    \textbf{Case 2: $|S\setminus \{x,y\}|\leq 1$.} As the second case, suppose that $|S\setminus \{x,y\}|\leq 1$. First, note that $|S\setminus \{x,y\}|= 0$ is not possible because this would mean that {$S = \{x,y\}$.}
    However, this implies that $\eta(y)=\lambda(x)=p$, which contradicts our assumption that $\lambda(x)\neq p$. Hence, $|S\setminus \{x,y\}|= 1$. Now, fix three agents $z_1,z_2,z_3\in N\setminus (S\cup \{x,y\})$ and their respective houses $q_i=\eta(z_i)$. Such agents exist since $|S\cup \{x,y\}|\leq 3$ and $n\geq 6$. Further, it holds that $\eta(z_i)=\lambda(z_i)$ for all $i\in \{1,2,3\}$ since $z_i\not\in S$. We will next analyze the preferences of agents $z_i$ and $z_j$ over the houses $q_i$ and $q_j$ for arbitrary $i,j\in \{1,2,3\}$ with $i\neq j$. 
    
    {Firstly}, assume that $q_i\succ_{z_i} q_j$ and $q_j\succ_{z_j} q_i$ (i.e., both agents prefer their obtained house to the house of the other agent). In this case, we define $\mu$ as the assignment obtained from $\lambda$ by swapping the houses of agents $z_i$ and $z_j$, and the houses of agent $x$ and the agent $y'$ with $\lambda(y')=p$. We note that, since $\lambda(z_k)=\eta(z_k)$ and $z_i\neq y$ for $k\in \{1,2,3\}$, it follows that $y'\not\in \{z_i,z_j\}$. Further, by definition, agents $z_i$ and $z_j$ prefer $\lambda$ to $\mu$. Since all agents but $z_i$, $z_j$, $x$, and $y'$ are indifferent between $\lambda$ and $\mu$, this suffices to show that $\lambda\mprefsim \mu$. Finally, since $\mu(x)=p$, this proves that $\lambda\in\tc(P)$.

    {Secondly}, suppose that $q_j\succ_{z_i} q_i$ and $q_i\succ_{z_j} q_j$ (i.e., both agents prefer the house of the other agent to their obtained house). In this case, we define the intermediate assignment $\eta'$ by letting $z_i,z_j$ swap their houses under $\eta$, i.e., by $\eta'(z_i)=q_j$, $\eta'(z_j)=q_i$, and $\eta'(w)=\eta(w)$ for all $w\in N\setminus \{z_i,z_j\}$. It holds that $\lambda\mprefsim \eta'$:
    
    \begin{enumerate}
        \item All agents agents in $S$, i.e., at least two agents, strictly prefer $\lambda$ to $\eta'$ because, by definition, they strictly prefer $\lambda$ to $\eta$, and are assigned the same house under $\eta$ and $\eta'$.
        \item Precisely two agents, namely $z_i,z_j$, strictly prefer $\eta'$ to $\lambda$ because, by definition, $z_i,z_j$ obtain the same house in $\lambda$ and $\eta$, whereas they strictly prefer $\eta'$ to $\eta$.
        \item All agents in $N\setminus (S\cup \{z_i,z_j\})$ are indifferent between {$\lambda$} and $\eta'$. Indeed, all agents other than $z_i,z_j$ are indifferent between {$\eta$} and $\eta'$, and $\lambda $ coincides with $\eta$ outside of $S$.
    \end{enumerate}

    Next, let $\mu$ be the profile derived from $\eta'$ by swapping the houses of $z_i$ and $z_j$ back, as well as the houses of $x$ and the agent $y'$ that has $p$ in $\eta'$, i.e., $\eta'(y')=p$.
Note that, since $\eta'$ is obtained from $\eta$ by $z_i,z_j$ swapping their houses, and $ y \not\in \{z_i,z_j\}$, we immediately obtain that $y'=y$. Therefore, $y'\neq x $ and $y'\not\in \{z_i,z_j\}$.
Since both $z_i$ and $z_j$ prefer their assignment under $\eta'$ to $\mu$, we get that $\eta'\mprefsim \mu$ and thus $\lambda\mprefsimt \mu$.

    {Thirdly}, by applying the previous analysis for every two-agent subset of $\{z_1,z_2,z_3\}$, we conclude that $\lambda\in\tc(P)$ unless
    \begin{align*}
        q_1\succ_{z_1} q_2 \text{ if and only if } q_1\succ_{z_2} q_2,\\ q_1\succ_{z_1} q_3 \text{ if and only if } q_1\succ_{z_3} q_3,\\ q_2\succ_{z_2} q_3 \text{ if and only if } q_2\succ_{z_3} q_3.
    \end{align*} 
    Now, let $\hat \lambda$ (resp. $\hat \eta$) denote the assignment derived from $\lambda$ (resp. $\eta$) by letting agents $z_1$ and $z_2$ swap their houses. Since $z_1$ and $z_2$ have the same preference regarding $q_1$ and $q_2$, it follows that $\lambda\mprefsim\hat \lambda$. Further, $\hat\lambda$ differs from $\hat\eta$ only in the assignments of the agents in $S$, and all the agents in this set strictly prefer their house under $\hat\lambda$. Using the same case analysis as in the last two cases for these two assignments, we can construct an assignment $\hat\mu$ such that agent $x$ obtains $p$ in $\hat \mu$ and $\lambda\mprefsim\hat\lambda\mprefsimt\hat\mu$ or the following equivalences hold: 
    \begin{align*}
        q_1\succ_{z_1} q_2 \text{ if and only if } q_1\succ_{z_2} q_2,\\ q_2\succ_{z_1} q_3 \text{ if and only if } q_2\succ_{z_3} q_3,\\ q_1\succ_{z_2} q_3 \text{ if and only if } q_1\succ_{z_3} q_3.
    \end{align*} 
    Further, repeating the argument for the assignments $\bar \lambda$ and $\bar \eta$ derived from $\lambda$ and $\eta$, respectively, by letting agents $z_2$ and $z_3$ swap their houses, we again get that there is either an assignment $\bar \mu$ such that $x$ obtains $p$ and $\lambda\mprefsim \bar\lambda\mprefsimt \bar \mu$, or that the following equivalences hold: 
    \begin{align*}
        q_1\succ_{z_1} q_3 \text{ if and only if } q_1\succ_{z_2} q_3,\\ q_1\succ_{z_1} q_2 \text{ if and only if } q_1\succ_{z_3} q_2,\\ q_2\succ_{z_2} q_3 \text{ if and only if } q_2\succ_{z_3} q_3.
    \end{align*} 

    Finally, if all $9$ given equivalences are true, the preferences of agents $z_1$, $z_2$, and $z_3$ over $q_1$, $q_2$, and $q_3$ have to coincide. Without loss of generality, we assume that $q_1\succ_z q_2\succ_z q_3$ for all $z\in \{z_1,z_2,z_3\}$. In this case, let $\mu$ denote the following assignment derived from $\eta$: agent $x$ and $y$ swap their houses and we set $\mu(z_1)=q_2$, $\mu(z_2)=q_3$, and $\mu(z_3)=q_1$. We claim that $\lambda\mprefsim \mu$. To see this, we note that $\lambda(z_1)\pref_{z_1}\mu(z_1)$, $\lambda(z_2)\pref_{z_2}\mu(z_2)$, and $\lambda(w)\pref_w \mu(w)$ for the single agent $w\in S\setminus \{x,y\}$. Since $\lambda$ and $\mu$ only differ in the houses of $6$ agents, this proves that $\lambda\mprefsim \mu$. Finally, we have constructed in every case a path from $\lambda$ to a suitable assignment $\mu$ in the top cycle, thus showing that $\lambda\in \tc(P)$ if $|S\setminus \{x,y\}|=1$.    
\end{proof}

Finally, based on \Cref{lem:notpso_subseteq_tc}, it is now easy to verify the remaining cases of \Cref{thm:Structure_TC}.

{\renewcommand{\footnote}[1]{}\TC*
}
\begin{proof}
Fix a profile $P$ for $n\geq 5$ agents. The cases (i) and (ii) follow directly from \Cref{lem:TC_Cardinality_1,lem:TC_Cardinality_2}. Since these lemmas characterize when, $|\tc(P)|\leq 2$, it follows that $|\tc(P)|>2$ for the remaining cases. By \Cref{lem:notpso_subseteq_tc}, we get that $M\setminus \pso(P)\subseteq \tc(P)$ and \Cref{thm:PO_Subset_SP} (combined with \Cref{lem:dualTCBC}) shows that $\pso(P)\subseteq \bc(P)$. Next, for cases (iii) to (v), we note that $\tc(P)=\bc(P)=M$ if $\tc(P)\cap \bc(P)\neq\emptyset$. In particular, if $\tc(P)\cap \bc(P)\neq\emptyset$, there is an assignment $p$ that can be reached and reaches every other assignment via a path in the majority relation, so any two assignments are connected via $p$. 

Now, under case (iii), we get from \Cref{lem:TC_Cardinality_2} that $|\bc(P)|=2$ and, moreover, it is easy to see that $|\pso(P)|=2$. This means that $\bc(P)\cap  \tc(P)= \emptyset$ since $|\tc(P)|\geq |M\setminus \pso(P)|\geq n!-2$. This is only possible if $\tc(P)=M\setminus \pso(P)$ and $\bc(P)=\pso(P)$, which completes the proof of this case.  

Similarly, under case (iv), we infer from \Cref{lem:TC_Cardinality_1} that $|\bc(P)|=1$. Further, it is straightforward to see that $|\pso(P)|=1$. Just as before, this means that $\bc(P)\cap  \tc(P)= \emptyset$, so  $\tc(P)=M\setminus \pso(P)$ and $\bc(P)=\pso(P)$. 

Finally, we turn to case (v). To this end, we recall that by \Cref{lem:TC_Cardinality_1,lem:TC_Cardinality_2}, both $|\tc(P)|>2$ and $|\bc(P)|>2$ if none of the first four cases apply. If $n=5$, our claim follows directly from Claim (3) of \Cref{fact:n5}. We hence suppose that $n\geq 6$. In this case, we derive from \Cref{lem:onetop} that there is an agent $x$ and a house $p$ such that $x$ top-ranks $p$ and every assignment $\mu$ with $\mu(x)=p$ is in the top cycle. If any such assignment $\lambda$ is not Pareto-optimal, then $\lambda\in \tc(P)$ by \Cref{lem:onetop} and $\lambda\in M\setminus \po(P)\subseteq\bc(P)$ by \Cref{lem:notpso_subseteq_tc}. Hence, $\tc(P)=\bc(P)=M$ by our previous observation. On the other hand, if every assignment $\mu$ with $\mu(x)=p$ is Pareto-optimal, then all agents but $x$ agree on the ranking of the houses in $H\setminus \{p\}$. Indeed, if there were houses $q,r\in H\setminus \{p\}$ and agents $y,z\in N\setminus \{x\}$ with $q\succ_y r$ and $r\succ_z q$, then every assignment $\lambda$ with $\lambda(y)=r$ and $\lambda(z)=q$ is dominated by the assignment where we swapped the houses of $y$ and $z$. 

Hence, suppose that all agents in $N\setminus \{x\}$ report the same preferences over $H\setminus \{p\}$. Without loss of generality, we label the houses in $H\setminus \{p\}$ such that $q_1\succ_y q_2\succ_y\dots\succ_y q_{n-1}$ for all agents $y\in N\setminus \{x\}$. In this case, we fix an arbitrary assignment $\lambda$ and show that $\lambda\in\tc(P)$. For this, let $y$ denote the agent with $\lambda(y)=p$, and label the agents in $N\setminus \{x,y\}$ by $z_1,\dots, z_{n-2}$ such that $\lambda(z_1)\succ_z \dots \succ_z\lambda(z_{n-2})$ for all $z\in N\setminus \{x,y\}$. Further, consider the assignment $\mu$ such that $\mu(x)=p$, $\mu(y)=\lambda(x)$, $\mu(z_{n-2})=\lambda(z_{1})$, and $\mu(z_{i})=\lambda(z_{i+1})$ for all $i\in \{1,\dots, n-3\}$. We note that all agents $z_i$ with $i\in \{1,\dots, n-3\}$ prefer $\lambda$ to $\mu$ because $\lambda(z_i)\succ_{z_i}\lambda(z_{i+1})=\mu(z_i)$. Since $n\geq 6$, we have that $n-3\geq\frac{n}{2}$, so $\lambda\mprefsim \mu$. Further, $\mu\in\tc(P)$ since $\mu(x)=p$. This shows that $\lambda\in\tc(P)$ and thus $\tc(P)=M$. 
\end{proof}

\section{Path for Example~\ref{ex:bad_in_tc}}\label{app:path_bad}
Here, we provide a concrete path from the assignment $\mu$ in \Cref{ex:bad_in_tc} to a serial dictatorship $\lambda$ in the majority graph. What follows is a sequence of pairs of assignments with the red assignment weakly dominating the other one (which is marked in blue for the agents receiving a different house). The profile $P$ remains the same. In the first illustration, the red assignment is $\mu$. In the last illustration, the blue assignment is $\lambda$.
\[\arraycolsep=3pt P = \begin{array}{rccccccc}
		1\colon & f,&b,&d,&e,&c,&a,&\badass g\\
		2\colon & \goodass d,&f,&g,&a,&b,&e,&\badass c\\
		3\colon & d,&a,&c,&g,&e,&b,&\badass f\\
		4\colon & a,&d,&b,&f,&g,&\badass e,&\goodass c\\
		5\colon & c,&g,&e,&b,&f,&\badass d,&\goodass a\\
		6\colon & f,&a,&e,&d,&g,&c,&\badass b\\
		7\colon & c,&d,&\goodass e,&b,&g,&f,&\badass a\\
	\end{array}
\]
\[\arraycolsep=3pt P = \begin{array}{rccccccc}
		1\colon & f,&b,&d,&e,&c,&a,&\badass g\\
		2\colon & \badass d,&\goodass f,&g,&a,&b,&e,&c\\
		3\colon & \goodass d,&a,&c,&g,&e,&b,&\badass f\\
		4\colon & a,&d,&b,&f,&g,&e,&\badass c\\
		5\colon & c,&g,&e,&b,&f,&d,&\badass a\\
		6\colon & f,&a,&\goodass e,&d,&g,&c,&\badass b\\
		7\colon & c,&d,&\badass e,&\goodass b,&g,&f,&a\\
	\end{array}
\]
\[\arraycolsep=3pt P = \begin{array}{rccccccc}
		1\colon & \goodass f,&b,&d,&e,&c,&a,&\badass g\\
		2\colon & d,&\badass f,&g,&a,&b,&\goodass e,&c\\
		3\colon & \badass d,&a,&c,&g,&e,&b,&f\\
		4\colon & \goodass a,&d,&b,&f,&g,&e,&\badass c\\
		5\colon & \goodass c,&g,&e,&b,&f,&d,&\badass a\\
		6\colon & f,&a,&\badass e,&d,&g,&c,&\goodass b\\
		7\colon & c,&d,&e,&\badass b,&\goodass g,&f,&a\\
	\end{array}
\]
\[\arraycolsep=3pt P = \begin{array}{rccccccc}
		1\colon & \badass f,&\goodass b,&d,&e,&c,&a,&g\\
		2\colon & d,&f,&g,&a,&b,&\badass e,&c\\
		3\colon & \badass d,&a,&c,&g,&e,&b,&f\\
		4\colon & \badass a,&d,&b,&f,&g,&e,&c\\
		5\colon & \badass c,&\goodass g,&e,&b,&f,&d,&a\\
		6\colon & \goodass f,&a,&e,&d,&g,&c,&\badass b\\
		7\colon & \goodass c,&d,&e,&b,&\badass g,&f,&a\\
	\end{array}
\]
\[\arraycolsep=3pt P = \begin{array}{rccccccc}
		1\colon & f,&\goodass b,&d,&e,&c,&a,&g\\
		2\colon & d,&f,&g,&a,&b,&\goodass e,&c\\
		3\colon & \goodass d,&a,&c,&g,&e,&b,&f\\
		4\colon & \goodass a,&d,&b,&f,&g,&e,&c\\
		5\colon & c,&\goodass g,&e,&b,&f,&d,&a\\
		6\colon & \goodass f,&a,&e,&d,&g,&c,&b\\
		7\colon & \goodass c,&d,&e,&b,&g,&f,&a\\
	\end{array}
\]
The assignment $\goodass\lambda$ is a serial dictatorship, e.g., with the picking order $3,4,6,7,1,5,2$.

\end{document}